\documentclass[english,11pt]{article}
\usepackage{amstext}
\usepackage{amsmath}
\usepackage{amsthm}
\usepackage{amssymb}
\usepackage{graphicx}
\usepackage{babel}
\usepackage{xcolor}
\usepackage{times}
\usepackage{xspace}
\usepackage{paralist}

\usepackage[top=1.in, bottom=1.in, left=1.in, right=1.in]{geometry}

\newcommand{\denselist}{\itemsep 0pt\parsep=0.8pt\partopsep 0pt}
\newcommand{\bitem}{\begin{itemize}\denselist}
\newcommand{\eitem}{\end{itemize}}
\newcommand{\benum}{\begin{enumerate}\denselist}
\newcommand{\eenum}{\end{enumerate}}

\newcommand{\E}[0]{{\mathbb{E} }}
\newcommand{\defn}[1]           {{\textit{\textbf{\boldmath #1\/}}}}
\newcommand{\Prob}[1]{\text{Prob}\left\{ #1 \right\}}
\newcommand{\Min}[1]{\text{min}\left\{ #1 \right\}}
\newcommand{\Po}[0]{{\mathbf{P}}}

\def\etal{\emph{et~al.}\xspace}

\DeclareGraphicsExtensions{.pdf,.png,.jpg}

\newtheorem{theorem}{Theorem}
\newtheorem{lemma}[theorem]{Lemma}
\newtheorem{claim}[theorem]{Claim}

\newtheorem{definition}{Definition}

\newtheorem{remark}{Remark}
 \newtheorem{corollary}[theorem]{Corollary}
  
\newenvironment{restate}[1]
{\par \vspace{1.5ex} \noindent \textbf{Statement of #1.}  \em}
{\par} 

\newcommand{\full}[1]{}               

\newcommand{\old}[1]{{}}

\def\E{{\cal E}}

\newlength\abovesectionskip
\newlength\belowsectionskip
\def\sectionfont{\normalfont\Large\bfseries}
\newlength\abovesubsectionskip
\newlength\belowsubsectionskip
\def\subsectionfont{\normalfont\large\bfseries}
\newlength\abovesubsubsectionskip
\newlength\belowsubsubsectionskip
\def\subsubsectionfont{\normalfont\normalsize\bfseries}
\newlength\aboveparagraphskip
\newlength\belowparagraphskip
\def\paragraphfont{\normalfont\normalsize\bfseries}
\makeatletter
\def\section{\@startsection{section}{1}{\z@}{-\abovesectionskip}%
              {\belowsectionskip}{\sectionfont}}
\def\subsection{\@startsection{subsection}{1}{\z@}{-\abovesubsectionskip}%
                 {\belowsubsectionskip}{\subsectionfont}}
\def\subsubsection{\@startsection{subsubsection}{1}{\z@}%
                    {-\abovesubsubsectionskip}{\belowsubsubsectionskip}%
                    {\subsubsectionfont}}
\def\paragraph{\@startsection{paragraph}{4}{\z@}{-\aboveparagraphskip}%
                {\belowparagraphskip}{\paragraphfont}}
\makeatother

\usepackage{times}

\def\compactify{\itemsep=0pt \topsep=0pt \partopsep=0pt \parsep=0pt}
\let\latexusecounter=\usecounter

\begin{document}
\begin{titlepage}

\title{How Complex Contagions Spread Quickly in the Preferential Attachment Model and Other Time-Evolving Networks}
\author{Roozbeh Ebrahimi\thanks{Department of Computer Science, Stony Brook University, Stony Brook, NY 11794. \texttt{\{rebrahimi,jgao,gghasemiesfe\}@cs.stonybrook.edu}} \and Jie Gao$^\ast$ \and Golnaz Ghasemiesfeh$^\ast$
\and Grant Schoenebeck\thanks{Department of Computer Science and Engineering, University of Michigan, Ann Arbor, Michigan, MI 48109. \texttt{schoeneb@umich.edu}}}

\clearpage\maketitle
\thispagestyle{empty}

\begin{abstract}
In this paper, we study the spreading speed of complex contagions in a social network. A $k$-complex contagion starts from a set of initially infected seeds such that any node with at least $k$ infected neighbors gets infected. Simple contagions, i.e., $k=1$, quickly spread to the entire network in small world graphs. However, fast spreading of complex contagions appears to be less likely and more delicate; the successful cases depend crucially on the network structure~\cite{G08,Ghasemiesfeh:2013:CCW}.

Our main result shows that complex contagions can spread fast in a general family of time-evolving networks that includes the preferential attachment model~\cite{barabasi99emergence}. We prove that if the initial seeds are chosen as the oldest nodes in a network
of this family, a $k$-complex contagion covers the entire network of $n$ nodes in $O(\log n)$ steps. We show that the choice of the initial seeds is crucial. If the initial seeds are uniformly randomly chosen in the PA model, even if we have a polynomial number of them, a complex contagion would stop prematurely. The oldest nodes in a preferential attachment model are likely to have high degrees. However, we remark that it is actually not the power law degree distribution per se that facilitates fast spreading of complex contagions, but rather the evolutionary graph structure of such models. Some members of the said family do not even have a power-law distribution. 

The main proof has two pillars. The first one is an analysis of a labeled branching process which might be of independent interest. The second pillar is an intricate coupling argument that links the extinction time of the labeled branching process to the speed of a $k$-complex contagion in the said family of time-evolving networks. The coupling argument itself relies on a careful revealing process that reveals the randomness of the network in a particular order to alleviate dependency/conditioning problems.

Using similar techniques, we also prove that complex contagions are fast in the copy model~\cite{KumarRaRa00}, a variant of the preferential attachment family, if the initial seeds are chosen as the oldest nodes.

Finally, we prove that when a complex contagion starts from an arbitrary set of initial seeds on a general graph, determining if the number of infected vertices is above a given threshold is $\Po$-complete. Thus, one cannot hope to categorize all the settings in which complex contagions percolate in a graph.

\bigskip \noindent  \textbf{keywords}: Social Networks, Complex Contagion, The Preferential Attachment Model, The Copy Model, Time-Evolving Networks.
\end{abstract}

\end{titlepage}



\section{Introduction}
Social behavior is undoubtedly one of the defining characteristics of us as a species. Social acts are influenced by the behavior of others while at same time influencing them. Understanding the dynamics of influence and modeling it in social networks is thus a key step in comprehending the emergence of new behaviors in societies. Similar to rumors or viruses, behavior changes manifest contagion like properties while spreading in a social network: Taking ``selfies'' started with a few people and suddenly became universal in a matter of months. Some of these contagions are beneficial (e.g., adopting healthy lifestyle) or profitable (e.g., viral marketing), while some others are destructive and undesirable (such as teenager smoking, alcohol abuse, or vandalism). To effectively promote desirable contagions and discourage undesirable ones, the first step is to understand how these contagions spread in networks and what are the important parameters that lead to fast spreading.

Our focus in this paper is on contagions that are \emph{complex}, contagions that require social reaffirmation from multiple neighbors, as opposed to \emph{simple} ones, which can spread through a single contact. Viruses or rumors can spread through a single contact and are thus adequately modeled by simple contagions. But when agents' actions and behavioral changes are involved, it has been argued in sociology literature that complex contagions represent most of the realistic settings -- making an important distinction between the \emph{acquisition} of information and the decision to \emph{act} on the information. While it takes only a single tie for people to hear about a new belief, technology, fad or fashion, ``it is when they see people they know getting involved, that they become most susceptible to recruitment'', as Centola and Macy~\cite{G08} explain. Many examples of complex contagions have been reported in social studies, including buying pricey technological innovations, changes in social behaviors, the decision to migrate, etc.~\cite{Coleman:1966, centola2010spread}. Studies of large scale data sets from online social networks have confirmed complex contagions as well. A study on Facebook discovered that having more than one friend already on Facebook who are not well connected to each other substantially increases the likelihood of one joining Facebook~\cite{ugander12}. A study on Twitter showed a similar phenomenon, that \emph{persistence} (the effect of repeated exposure to a topic) plays an important role in the diffusion of hashtags~\cite{Romero11}.

Simple contagions and epidemics have been extensively studied (ref. to~\cite{DK10, Jackson08, Newman10}).  Simple contagions can spread fast in social networks because these networks typically have the small world property. A single tie could leap over large network distances and spread the epidemic to a remote community. In contrast, fast spreading of complex contagions appears to be much more delicate and difficult. Preliminary research by~\cite{G08} and recent work by~\cite{Ghasemiesfeh:2013:CCW} show that for a number of small world models, in which simple contagions are super fast, complex contagions are \emph{exponentially} slower. Despite the crucial importance of complex contagions in accurately modeling a wide range of social behaviors, besides the above results, their diffusion behavior hasn't been rigorously studied much. This is possibly due to the difficulty of formal analysis of complex contagions. The difficulty arises in two aspects. First, the required multiple infections mean that subsequent exposures do not always have diminishing returns which turns out to be mathematically challenging to handle.  For example, it violates submodularity, and even subadditivity, on which many analyses depend.  Second, the superadditive character of  complex contagions means that they are integrally related to community structure, as complex contagions intuitively spread better in dense regions of a network~\cite{centola2010spread}.

We adopt a model of contagion called $k$-complex contagion from \cite{G08,Ghasemiesfeh:2013:CCW}.
A $k$-complex contagion starts from a set of initially infected seeds, and any node with at least $k$ infected neighbors gets infected. While being simple, this model elegantly captures the core difference between complex and simple contagions~\cite{G08, centola2010spread}, and despite simplicity, it is already difficult to analyze. In a clique, a $k$-complex contagion immediately infects every node as long as there are $k$ seeds.  In contrast to the dense clique, our work studies complex contagions on sparse networks, in which the average degree is constant.  We demonstrate that this model yields important contributions to the understanding of the role of network structure in social contagions.


This work seeks to enrich our understanding of complex contagions by answering fundamental questions on time-evolving graphs, which often have a power-law distribution. Additionally, this paper develops theoretical tools that enable us to overcome some of the challenges in understanding complex contagions.

We study two models of time-evolving networks. The first is the preferential attachment model, which is one of the most studied generative models with a power law degree distribution. Price in 1976 attributed the appearance of power law degree distributions to the mechanism of ``cumulative advantage'', now more commonly known as preferential attachment, phrased by Barabasi~\cite{barabasi99emergence}. The preferential attachment model considers an evolving network in which newcomers link to nodes already in the network with a probability proportional to the current degree of the nodes. Thus, nodes that have an advantage over the others in terms of degree will attract more links as the network evolves. In a slightly different model by Kleinberg~\etal~\cite{Kleinberg:1999:WGM} and Kumar~\etal~\cite{Kumar:1999:ELK,KumarRaRa00} (now called  the ``copy model''), a newcomer chooses its links uniformly randomly from existing nodes with a small probability $p$, and with probability $1-p$ copies the links of a prototype node\footnote{Which is also uniformly randomly chosen.}. Also in this setting, the newcomer chooses a node with probability proportional to its degree. Thus, the same ``rich-get-richer'' mechanism leads to a power law degree distribution. Other variations of the PA model were proposed in~\cite{fitness-model,Lattanzi:2009:AN, Pennock02winnersdont,papadopoulos12popularity}. Refer to~\cite{Mitzenmacher_abrief} for a nice survey of the history of PA models.

The power law degree distribution of these graphs means that, the number of nodes having degree $d$ is proportional to $1/d^\gamma$, for a positive constant $\gamma$.  In 1965, Derek de Solla Price showed that the number of citations to papers follow a power law distribution. Later, a number of papers studying the WWW reported that the network of webpages also has a power law degree distribution~\cite{barabasi99emergence,Broder:2000}. Besides the ubiquitous observations of power law distributions in social networks, many other networks such as biological, economic and even semantic networks were shown to have power law degree distributions~\cite{SteyversTenenbaum05,reka02statistical,Newman03thestructure} as well.

Because these time-evolving networks typically have power-law distributions, not all the nodes are homogeneous.  This mirrors reality in that people may be very different in how influential they are. They differ not only in their personal traits such as leadership, charisma, etc., but also in the positions they take in the social network. A number of previous works acknowledge such differences and compute the `network' value of a user, as the expected profit from sales to other customers she may influence to buy, the customers those may influence, and so on~\cite{Domingos:2001:MNV}.  This heterogeneity allows us to study the effect of nodes which are initially infected, which is an aspect of complex contagions not examined in previous theoretical work.


\smallskip\noindent\textbf{Our results.}
The main result of this paper is to show that complex contagions can spread fast in a general family of time-evolving networks that includes the \emph{independent}, the \emph{sequential} and the \emph{conditional} preferential attachment models~\cite{reka02statistical, Bol04,BergerBoCh10}. We prove that if the initial seeds are the oldest nodes in a network of this family, a $k$-complex contagion covers the entire network of $n$ nodes in $O(\log n)$ steps. This is surprising because these networks do not contain any community structure, per se, yet they still support fast spreading of complex contagions. Using similar techniques, we also prove the same result for the copy model~\cite{KumarRaRa00}.
%

For the preferential attachment model, when the probability of creating edges using the preferential attachment rule, $p$, is in $[0, 1)$ (ref. to Definition~\ref{def:PA}) we conjecture that w.h.p. the diameter is $\Theta(\log n)$, and thus our result is tight up to a constant factor\footnote{Dommers~\etal~\cite{Dommers} show that, if the exponent of the power-law distribution is greater than 3, then the PA model has a diameter  of $\Theta(\log n)$.  Berger~\etal~\cite{Berger05} prove that if $p \in [0, 1)$ in Definition~\ref{def:PA}, then the exponent of the power-law distribution is greater than 3.  However, while Berger~\etal use the same PA model as Definition~\ref{def:PA}, the model in Dommers~\etal is slightly different. It is beyond the scope of this paper to extend the results of Dommers~\etal to this setting, but we know of no barriers to doing so.}. This means that, if the initial seeds are properly chosen, the speed of simple and complex contagions differ only by a constant factor.
When $p = 1$, it is known that the diameter is $\Theta(\log n/\log \log n)~$\cite{Bol03,Berger05}, and so in this setting complex contagions are at most a $\log \log n$ factor slower than simple contagions.

We also show that the choice of the initial seeds is crucial. We show that there exists a polynomial threshold $f(n)$ such that if $o(f(n))$ initial seeds are chosen uniformly at random in the PA model, the contagion almost surely does \emph{not} spread! Second, we show that if $\Omega(f(n)\log n)$ initial seeds are infected, the oldest nodes and then the whole graph gets infected w.h.p. in $O(\log n)$ rounds. This signifies not only the importance of the choice of initial seeds, but also the delicacy of the diffusion for a complex contagion.

The oldest nodes in a preferential attachment model are likely to have high degrees.  However, we remark that it is actually not the power law degree distribution per se that facilitates the spread of complex contagions, but rather the evolutionary graph structure of such models. Indeed, the time-evolving network family also includes heavily concentrated degree distributions with the largest degree being $O(\log n)$. 

%
%
%

While one might hope to categorize all the settings in which complex contagions spread, we show that this is unlikely.  We prove that given a graph, a list of initially infected nodes, and a threshold, it is $\Po$-complete to decide if the number of infected nodes surpasses the threshold or not.  Thus, in some sense, the best one can do (in the worst-case) is to simulate the contagion.

\smallskip\noindent\textbf{Organization of this paper.}
In Section~\ref{sec:related}, we outline some of the related works to our paper.
Section~\ref{sec:prelim} contains the preliminary definitions and models.
Section~\ref{sec:time-evolving} contains the main result of this paper. Due to the technical nature of the result,
we provide a proof overview of it in Subsection~\ref{sec:proof-overview} first.
Subections~\ref{sec:branching}, \ref{sec:staging}, and \ref{sec:arrival-prop-speed} contain the technical details of the proof
of the main result (Theorem~\ref{thm:fast-contagion-staged} and Corollary~\ref{thm:main}).
In Appendix~\ref{app:copy}, we prove an analogous result to the main theorem about the copy model.
Appendix~\ref{sec:bootstrap} addresses the random choice of initial seeds for a complex contagion in the preferential attachment graph.
Finally, Appendix~\ref{sec:complexity} proves that computing the extent of complex contagions in general graphs is $\Po$-complete.


\section{Related Work}
\label{sec:related}
Diffusion of information/viruses has been an active research topic in epidemics, economics, and computer science. For a complete overview, refer to the references in~\cite{Ghasemiesfeh:2013:CCW}. We describe the most relevant results here.

First, we remark that our model of complex contagions belongs to the general family of \emph{threshold models} in the study of diffusions. In the \emph{threshold model}, each node has a threshold on the number of active edges/neighbors needed to become activated~\cite{Gran78} (In a $k$-complex contagion, all the nodes have the same threshold $k$).  The threshold model is motivated by certain coordination games studied in the economics literature in which a user maximizes its payoff when adopting the behavior as the majority of its neighbors.  Many of the studies focus on the stable states, and structural properties that prevent complete adoption of the advanced technology (better behaviors)~\cite{Morris97contagion}. Montanari and Sabari~\cite{Mon09} is among the few studies that relate the  steady state convergence speed of the coordination game to the network structure.

Diffusion of simple contagions in preferential attachment models has been extensively studied. In \cite{Berger05}, Berger~\etal~studied the spread of viruses where the underlying graph was considered as a PA model.  Chierichetti~\etal~\cite{Chier11} studied the spread of rumors under the \emph{push-pull strategy} model on PA graphs. Later in \cite{Chier10} and \cite{Chieri10} they improved their bound and also made a relation between the spread of rumors and conductance of a graph.  Recently, Doerr~\etal~\cite{Doerr11} proved the tighter bound on the diffusion of rumors in the PA model where the model of diffusion is a slight variation of push-pull strategy. The study of complex contagions and their speed in a preferential attachment model, as in this paper, is new.

There are a number of empirical studies on the diffusion in networks in general and role/attributes of influential nodes~\cite{Bakshy09, Adar05, Aral09, Leskovec07}. Most of the studies related to ours  examine influence on Twitter. For examples, in~\cite{cha2010},~Cha~\etal. compare three different measures of influence in Tweeter: number of followers, number of retweets, and number of mentions. Among their interesting observations, one is that ``popular users who have high in-degree are not necessarily influential in terms of spawning retweets or mentions".  In another study of influence on Tweeter, Bakshy~\etal~\cite{Bakshy11} found out that users who have been influential in the past and have a large in-degree would generate the largest cascades. More can be found in~\cite{Weng10, Romero11, Kwak10}.

In \emph{bootstrap percolation}~\cite{Chalupa79bootstrap,adler91bookstrap}, all nodes have the same threshold but initial seeds are randomly chosen. Here, the focus is to examine the threshold of the number of initial seeds with which the infection eventually `percolates', i.e. diffuses to the entire network.
Janson~\etal~\cite{Janson12} examined the bootstrap percolation process on the random Erdos-Renyi graph, $G(n,p)$, for a complete range of parameters. Among their findings, they show that when the average degree of the network is constant ($p=d/n$ for $d=O(1)$), and the size of initial seeds is $o(n)$, the process will not cover a significant part of the network. Bootstrap percolation on random regular graphs~\cite{BaloghP07}, and the configuration model~\cite{Amini10} has been shown to follow a similar pattern as $G(n,p)$.
Recently, Amini and Fountoulakis~\cite{Amini12} showed a different pattern of percolation on inhomogeneous random graphs with power-law distributions.
They show that there exists a function $\alpha(n)=o(n)$ such that if the number of initial seeds is $\ll \alpha(n)$, the process does not evolve with high probability. If, on the other hand, the number of initial seeds is $ \gg \alpha(n)$, then a constant fraction of the graph is infected w.h.p.


\section{Preliminaries}\label{sec:prelim}
First, we formally define a $k$-complex contagion process in an undirected graph. We assume $k=O(1)$.
\begin{definition}
A \defn{$k$-complex contagion} ${\rm CC}(G, k, \mathcal{I})$ is a contagion that initially infects vertices of $\mathcal{I}$ and spreads over graph $G$. The contagion proceeds in rounds. At each round, each vertex with at least $k$ infected neighbors becomes infected. The vertices of $\mathcal{I}$ are called the initial seeds.
\end{definition}

There are a number of different definitions of the preferential attachment model, in which the difference lies in the subtle ways that the links are created. We mainly work with the \emph{independent model}~\cite{reka02statistical}.

\begin{definition}\label{def:PA}
The \defn{independent preferential attachment model}, ${\rm PA}_{p,m}(n)$: We start with a complete graph on $m+1$ nodes. At each subsequent time step $t = m+2, \cdots, n$ a node $v$ arrives and adds $m$ edges to the existing vertices in the network. Denote the graph containing the first $n-1$ nodes as $G_{n-1}$. For each new vertex, we choose $w_1, w_2, \cdots, w_m$ vertices, possibly with repetitions from the existing vertices in the graph. Specifically, nodes $w_1, w_2, \cdots, w_m$ are chosen independently of each other conditioned on the past. For each $i$, with probability $p$, $w_i$ is selected from the set of vertices of $G_{n-1}$ with probability proportional to the vertices' degree in $G_{n-1}$; and with probability $1-p$, $w_i$ is selected uniformly at random. Then we draw edges between the new vertex and the $w_i$'s. Repeated  $w_i$'s cause multiple edges. Note that $deg(G_n)=2mn$. 
\end{definition}

There are two other variations of the PA model. In the \emph{conditional} model~\cite{BergerBoCh10}, a new edge is chosen conditioned on it being different from the other edges already built; in the \emph{sequential} model~\cite{Bol04}, the $m$ edges of the new node $v$ are built sequentially in the sense that the $i$-th edge of $v$ is chosen preferentially assuming the previous $i-1$ edges of $v$ have been included in the graph and their degrees are counted.
\begin{remark}
Our results also hold for the \emph{sequential} and \emph{conditional} PA models unless stated otherwise.
\end{remark}

A close relative of the ${\rm PA}$ model is the \emph{independent} \defn{copy model} of Kumar \etal \cite{KumarRaRa00}. 

\begin{definition} \label{def:CM}
The \defn{copy model}, ${\rm CM}_{p,m}(n)$, is generated as follows.  Initially, we start with a complete graph on $m+1$ vertices.  At each subsequent time step $t = m+2, \cdots, n$ a node $v$ arrives and adds $m$ edges to the existing vertices, $w_1, w_2, \cdots, w_m$, with possibly repetitions. First, the node chooses a prototype node $z$ uniformly at random. Then with probability $1-p$, $w_i$ is selected uniformly at random from the first $t-1$ vertices; and with probability $p$, $w_i$ is selected as the $i$-th outgoing neighbor of the prototype node $z$.
The $m$ edges are chosen independent of each other and hence there could be possibly multi-edges.
\end{definition}

In the \emph{conditioned} copy model, which avoids multi-edges and self-loops, all the edges are conditioned on them being different from each other. That is, the uniform random choice of the $i$-th edge is performed \emph{without} replacement. And in the case when the edge is copied from a prototype and the copied edge is already chosen, other indices of the out-going neighbors of $z$ are tried until success.
\begin{remark}
The results we prove for the independent CM model also hold for the conditioned CM model.
\end{remark}


\section{Complex Contagions in Families of Time-Evolving Networks} \label{sec:time-evolving}
In this section we prove that when initial seeds are chosen as the oldest $k$ nodes, $k$-complex contagions in
a family of time-evolving networks infect every node in $O(\log n)$ rounds. This family includes
all the variants of the preferential attachment graph.
We provide a proof overview before diving into the technical details.


\subsection{Challenges and Proof Overview}\label{sec:proof-overview}
Let $D$ be a graph created according to the PA model (Definition~\ref{def:PA}).
First, let us sketch a proof for $k = 1$, i.e. that with high probability $D$ has diameter $O(\log n)$.  Then we show where the analogous proof runs into trouble for $k > 1$.  This will motivate the machinery that we develop.

Label the vertices $1$, $2$, $3$, etc. according to their order of arrival.  We sketch a proof that the distance from an arbitrary node $v$ to vertex $1$ is $O(\log n)$ w.h.p. and the result follows from a union bound.

Consider the following procedure: 
\begin{inparaenum}[\bf a)]
\item  Start at $v$; 
\item  Follow the edge out of $v$ whose end point, $u$ has the lowest label; 
\item  If the label of $u$ is 1, stop. Otherwise, repeat the procedure for node $u$.
\end{inparaenum}

We claim that this procedure terminates in $O(\log n)$ steps with high probability.  Consider that at some point, the process is at vertex $u$.   Consider the induced subgraph on the vertices $\{1, 2, \ldots, u\}$.  If we have no prior knowledge, then it is easy to show that the lowest labelled neighbor of $u$ will be, in expectation, at most $\alpha u$ for some $\alpha < 1$.  The result follows from standard concentration arguments. 

However, the process does have knowledge of the graph when a vertex $u$ is processed. Namely, it knows the neighbors of all the vertices it has previously processed!  Fortunately, it is not too hard to show that if all these endpoints have indices greater than $u$, then the marginal distribution of edges on the induced subgraph of vertices $\{1, 2, \ldots, u\}$ remains unchanged.

Things go awry when we let $k = 2$. The first problem is that we need better concentration to be able to handle many nodes at the same time.  With $k = 1$, if we get unlucky and the first few steps did not move backward much from $v$, we are still doing at least as well as when we started.  However, when $k = 2$ and the first $\ell$ steps did not move backward much, we have $2^{\ell}$ vertices to process which is a  problem when $\ell=\Omega(1)$.

One idea of handling this is to partition the graph into stages.  Let stage $0$ contain the first $k$ vertices, while stage $i$ contains the vertices labeled between  $k(1 + \epsilon)^{i-1}$ and  $k(1 + \epsilon)^{i}$.  Thus, each stage will have a $(1 + \epsilon)$ fraction more vertices than the last.  The probability that a vertex in stage $i$ does not connect to $k$ vertices in previous stages can be upper bounded by a constant that depends on $k$ and $\epsilon$ and thus can be made arbitrarily small.  We can show that it takes at most an (expected) constant number of steps to get from one stage to the previous stages.  While this is sufficient for the proof to work in the case of $k = 1$, it is not enough for the cases of $k \geq 2$.   The reason is that only knowing the expectation does not give a tight enough bound when we process many vertices.  We need to bound the maximum rather than just the average.

To solve this problem, we model the above process as a  \defn{labeled branching process}, introduced in Subsection~\ref{sec:branching}, Definition~\ref{def:branching-process}.  A  branching process is a Markov process modeling a population where individuals in generation $i$ produce some number of individuals in generation $i+1$ according to a probability distribution. In a \defn{labeled branching process}, each individual has a label, and the probability distribution of producing an offspring is dependent on the labels of the parent/offspring.  

We intend to couple the random process that creates $D$ with a labeled branching process $B$. The labels in $B$ are proxies of the stages of nodes in $D$. After the coupling, the height of $D$ is bounded by the extinction time of $B$.   We use a potential function argument to study the extinction time of the labeled branching process.  We show that with high probability, the populations becomes extinct in $O(\log n)$ generations.
The coupling argument must make correspondence between the nodes/edges in $D$ and
nodes/branches in $B$ and thus relies on showing that the marginal probabilities of creating edges
in $B$ and in $D$ match.

The edges of $D$ are created in the arrival order of a PA graph (according to Definition~\ref{def:PA}).
However, $B$ reveals nodes/edges from last to first.
That is, the root branches (edges) are the first edges to be revealed in a branching process
and the root corresponds to the node labeled $n$.
Therefore, the coupling argument should follow a revealing process that processes
nodes in the reverse arrival order of the PA graph.

Unfortunately, at this point, more subtle problems arise.  With $k= 2$, we introduce new dependencies. Say we are processing the $100$-th arriving vertex, which has neighbors with arriving orders $33$ and $50$.  Then when we go to process vertex $50$, we have information about vertex $33$ --namely that it connects to vertex $100$.  In general, we are processing a node $u$, but the process has already revealed many outgoing edges from nodes $\{w\}_{w>u}$ to a node $s_{s<u}$, then the outgoing edges of $u$ are more likely to be connected to $s$ in the PA graph conditioned on the information revealed so far.  In contrast, in the arrival order of the PA graph, at the time $u$ created its edges, $s$ might not have had a high degree and thus the edges of $u$ would not be likely to be connected to $s$.
This ruins the above approach.  To rectify things, we need to be very careful about the order in which the edges are revealed.

Instead of revealing the neighbors of a particular vertex we query if individual edges (e.g. $(u, v)$) exist in the graph.  By the end, we have queried all the edges, but we do so in a carefully chosen order.  We do not ``process $v$" any more. Instead this ordering processes two edge points at a time.  However, when we process an edge $(u,v)$ we are able to relate the probability that this edge exists to a probability that it is created in a more natural revelation ordering (similar to the definition of PA).

Subsection~\ref{sec:staging} rigorously defines a \defn{revealing process} according
to an ordering (Definition~\ref{def:revealing}). We also introduce the revelation ordering we use on the edges (Definition~\ref{def:BF-order}: \defn{backward-forward (BF) ordering}) in this subsection. 
We show that the marginal probabilities of the individual edges conditioned on the information revealed in the \emph{backward-forward ordering} match the marginal probabilities of these edges in the arrival ordering of the PA graph.
This alleviates the dependency/conditioning problems described above.

Next, we show that the PA graph satisfies a staging property (Definition~\ref{def:staged}) which
roughly follows the staging described above: We can divide the PA graph into stages such that when edges are revealed
in the BF ordering, the probability that a node in stage $i$ does \emph{not} make
an edge to stage $i-1$ is bounded.

In Subsection~\ref{sec:arrival-prop-speed}, we show that the length of the longest
path from a node $u$ to node $1$ in staged graphs can be coupled to the extinction time of a labeled branching process $B$ (Theorem~\ref{thm:fast-contagion-staged}). One additional challenge is that the same vertex may repeat in the branching process. When this happens, we lose independence; however we show that by disallowing any children from all but the deepest labeled individual in $B$ corresponding to a particular vertex in $D$, we maintain independence without changing the height of $D$.   
 
We then conclude that the speed of a $k$-complex contagion on PA model is $O(\log n)$ if
the initial seeds are chosen as the first $k$ nodes in the graph (Corollary~\ref{thm:main}).

\subsection{Labeled Branching Processes}\label{sec:branching}
In this subsection, we describe one of our main tools in analyzing the speed of complex
contagions on time-evolving graphs. We define a \defn{labeled branching process} and analyze its extinction time.

\begin{definition}\label{def:branching-process}
For constants $m$ and $0< \alpha \leq 1$, we call a branching process a \defn{ $B(m, x, \alpha)$-labeled branching process}, if
\begin{inparaenum}
\item it starts with one node (root) labeled $x$ at depth 0 (where $x$ is a positive integer);
\item at each subsequent depth, every $i$-labeled node (where $i \neq 0$) produces $m$ children, and in expectation $\alpha m$ of the children have label $i-1$ and the rest have label $i$;
\item $0$-labeled nodes produce no children.
\end{inparaenum}
\end{definition}

The following lemma bounds the extinction time of a labeled branching process by $O(\log n)$, when there are $x=O(\log n)$ labels and $\alpha$ is a constant satisfying $\alpha >1- 1/m$.

\begin{lemma}  \label{lemma:main}
If $\alpha >1- 1/m$, and $x=c_1 \log n$ for a constant $c_1$, then the probability that $B(m, x, \alpha)$  has not died out after depth $t= c_2\log(n)$ is at most $n^{-(c_3+1)}$, where
 $c_3$ is a constant, $c_2 =  (c_3 + 1 + c_1/\log_{md}(e))/\log(1/\delta)$, $d = m\alpha/(1 - m(1 - \alpha))$, and $\delta =  m(1-\alpha)+1/m-(1-\alpha)$.
\end{lemma}

\begin{proof}
We refer to a node as an \defn{$(i-1)$-labeled origin} if it is $(i-1)$-labeled but its parents are not. Let $d$ be the expected number of $(i-1)$-labeled origin descendants of an $i$-labeled node $v$. First note that $d$ does not depend on $i$.  Clearly, any $(i-1)$-labeled children of $v$ are  $(i-1)$-labeled origins, and any $i$-labeled children of $v$ will produce in expectation $d$ descendants that are $(i-1)$-labeled origins. This gives us the equation that  $d =  m\alpha + m(1-\alpha)d$.  Assuming that $\alpha  > 1-1/m$ and solving, we find $d = m\alpha/(1 - m(1 - \alpha))$.
Then by independence, the expected number of $0$-labeled leaves of the root of the branching process is $d^x$.

We define a potential function $\phi(t)$ on the branching process $B$ at time $t$. Let $N_t(j)$ be the number of $j$-labeled nodes of $B$ at depth $t$.
Note that $N_0(x)=1$, and $N_0(j)=0$ for $0\leq j \leq x-1$.
Let
$$\phi(t) = \sum_{j = 1}^x N_t(j)  (md)^j.$$
We can verify that $\phi(0)$ is a polynomial in $n$, because
$\phi(0) = (md)^x = (md)^{c_1 \log n} = n^{c_1/\log_{md}(e)}$.
Next, we show that this potential function decreases exponentially with the time.
\begin{claim} \label{claim:martingale}
$\E[\phi(t+1)| \phi(t)] \leq \delta \phi(t)$, where $\delta = m(1-\alpha)+1/m-(1-\alpha)$.
\end{claim}
\begin{proof}
At level $t$, a node $v$ of label $i$ contributes $(md)^i$ to $\phi(t)$ for depth $t$. $v$'s contribution to $\phi(t+1)$ at depth $t+1$ is at most $m(\alpha(md)^{i-1} + (1-\alpha)(md)^i)$ in expectation.
We factor $(md)^{i}$ out, insert the value for $d$ from above and simplify to get $\delta$.
Notice that as long as $\alpha > 1 - 1/m$ we have that $\delta < 1$.
\end{proof}

Applying the previous claim allows us to prove by induction that $\E[\phi(t)] < \delta^t \phi(0)$.  Let $c_2 =  (c_3 + 1 + c_1/\log_{md}(e))/\log(1/\delta) $.  Then $\E[\phi(c_2 \log n)] = \delta^{c_2 \log n}\phi(0) < n^{-(c_3 + 1)}$.  If a node at time $t=c_2\log n$ existed it would contribute at least $(md)^1 \geq 1$ to $\phi$.   Thus, by Markov's inequality, we conclude that the probability that there are any nodes on the level $t$ is at most $n^{-(c_3 + 1)}$.
\end{proof}

Our notion of \emph{labeled branching process} is closely related to the notion of \emph{multitype Galton-Watson branching processes}
in the Markov process literature~\cite{Harris64}. Although the extinction time of multitype processes have been studied before~\cite{Harris64}, this literature has not explored the extinction time when the number of
types in the process is not a constant.
In our setting however, the number of types (labels) is $\Omega(1)$ and Lemma~\ref{lemma:main}
can be generalized to any number of labels bigger than $\log n$ with slight modification. In this sense, Lemma~\ref{lemma:main}
might be useful in its own right in multitype Galton-Watson branching processes theory. 

\subsection{Revealing Processes and the Staging Property}\label{sec:staging}
In this subsection, we define a \defn{staging property} notion and show that the PA model introduced in Section~\ref{sec:prelim} satisfies this property. Later we prove that a complex contagion is fast on graphs with the staging property
if it starts from the earliest nodes. The copy model, however, does not satisfy this property due to an inherent correlation between different outgoing edges of a node that come from a prototype node.

Let $\mathcal{G}$ be a distribution of graphs that is defined by a graph generation process over time.
\begin{definition}  \label{def:mgen}
    We will say that distribution $\mathcal{G}$ \defn{$m$-generates} a graph over time if:
    \begin{inparaenum}[i)]
       \item The process $\mathcal{G}$ starts with a clique at time $0$. At each time step at most one vertex arrives. The $i$-th arriving node is labeled index $i$.
       \item Each arriving vertex $v$ has at least $m$ edges to previously added vertices\footnote{These edges are possibly generated in a randomized way.}. For each edge $v \rightarrow u$, $u<v$.
    \end{inparaenum}
\end{definition}

\begin{definition}
Let $V$ be the set of vertices in an $m$-generated graph $G$,  and let $u, v \in V$, $j \in [m]$. We say that an \defn{ordered triple $(u, v, j)$} is \defn{oriented}, if $u < v$ in $G$'s arrival order.
\end{definition}
An \emph{oriented triple $(u, v, j)$} corresponds to the $j$-th edge that could be (potentially) issued by node $v$ to $u$ in the (randomly) generated graph. 

\begin{definition}
We define an \defn{arrival-time (AT) ordering} on triples as follows:
$(u_1, v_1, j_1) < (u_2, v_2, j_2)$ if
\begin{inparaenum}[a)]
\item $v_1 < v_2$ or;
\item if $v_1 = v_2$ and $j_1 < j_2$ or;
\item if $v_1 = v_2$ and $j_1 = j_2$ and $u_1 > u_2$.
\end{inparaenum}
\end{definition}
The \emph{AT ordering} is a sequential ordering of the edges that corresponds to the order that they are built in the evolving graph $G$.
That is, a node that arrives earlier will have its edges placed earlier. For the edges placed by the same node $v$, we sort them according to the inverse arriving order of their tails.

\begin{definition}\label{def:revealing}
Given an $m$-generative model $\mathcal{G}$ and an ordering $\mathcal{O}$, we define a \defn{revealing process $R_{\mathcal{O}}(\mathcal{G})$}.
We process all the oriented triples according to $\mathcal{O}$.  When processing a triple $(u, v, j)$, we reveal if the $j$-th edge from $v$ connects to $u$.  Let \defn{$\psi_{(u, v, j)}$} be the indicator r.v. for this event. Also, let \defn{$\phi(W, v, j)$} be the event that the $j$-th edge of $v$ lands in a set $W$ (where for all $u \in W$, $u < v$)\footnote{Note that $\phi(W, v, j)=\bigvee_{u \in W; u< v} \psi_{(u, v, j)}$.}.

When the first triple corresponding to an outgoing edge of a node is visited in the
ordering $\mathcal{O}$, the filter reveals the random choices of the generative model $\mathcal{G}$ specific to the node itself (not the edge choices). 

Let $L$ be a graph generated from $\mathcal{G}$ but with the $j$-th edge issued by vertex $v$ missing.  We define \defn{$p_{(u, v, j), L, \mathcal{O}}$} as the probability that $\psi_{(u, v, j)}$ occurs in $R_{\mathcal{O}}(\mathcal{G})$ when the triple $(u, v, j)$ is processed conditioned on the fact that the edges revealed thus far are consistent with $L$\footnote{Note that the revealing process doesn't know the edges of $L$ yet to be revealed, and so this probability is independent of that.}.  Define \defn{$p_{(W, v, j), L, \mathcal{O}}$} analogously.

We define a \defn{coin, $c_{(u, v, j), L, \mathcal{O}}$} to be a uniformly distributed r.v. in the interval $[0, 1]$.  We use $c_{(u, v, j), L, \mathcal{O}}$ to determine the event $\psi_{(u, v, j)}$ in the revealing process $R_{\mathcal{O}}(\mathcal{G})$ conditioned on the fact that the information revealed thus far is consistent with $L$.   If  $c_{(u, v, j), L, \mathcal{O}} \leq p_{(u, v, j), L, \mathcal{O}}$  then  $\psi_{(u, v, j)}$ \defn{occurs}, and o.w. \defn{it does not}.
\end{definition}


\begin{definition}\label{def:staged}
 Let $\mathcal{G}$ be an $m$-generative model and $R_{\mathcal{O}}(\mathcal{G})$ be a revealing process.
Let $G$ be any graph of size $n$ generated from $\mathcal{G}$.
We say that $\mathcal{G}$ satisfies the \defn{$(R_{\mathcal{O}},r, m, \alpha)$-staging property} if
there exists an ordering on the vertices of $G$ and an ordered partition $S_0, S_1, \ldots, S_r$ of the nodes into $r+1$ stages (the nodes in stage $i$ are ordered before those of $i + 1$) such that:
\begin{enumerate}[i)]\denselist
    \item $|S_0| < \log(n)$; 
    \item Each vertex has $m$ edges to nodes prior in the ordering;
    \item Assume that node $v$ is in stage $i$.  Let $W$ be the set of nodes in stage $i$ that precede $v$.  Let $L$ be any graph generated from $\mathcal{G}$ but with the $j$-th edge issued by  vertex $v$ missing. Then $ p_{(W, v, j), L, \mathcal{O}} \leq (1 - \alpha)$.
\end{enumerate}

A graph $H$ generated by a model $\mathcal{G}$ with staging property is said to be \defn{$(r, m, \alpha)$-staged}.
\end{definition}

The \defn{backward-forward ordering} sorts the oriented triples by the decreasing order of the landing vertices, and for nodes with the same landing vertices sorts them by the increasing order of the shooting vertices.

\begin{definition}\label{def:BF-order}
We define a \defn{backward-forward (BF) ordering} on triples as follows:
$(u_1, v_1, j_1) < (u_2, v_2, j_2)$ if
\begin{inparaenum}[a)]
\item $u_1 > u_2$ or;
\item if $u_1 = u_2$ and $v_1 < v_2$ or;
\item if $u_1 = u_2$ and $v_1 = v_2$ and $j_1 < j_2$.
\end{inparaenum}
\end{definition}

The \emph{BF order} is an interesting ordering for us because
of two reasons:
\begin{inparaenum}[a)]
\item It processes the nodes in the reverse arrival order and facilitates the coupling argument of Subsection~\ref{sec:arrival-prop-speed}; and
\item We can prove that the preferential attachment model satisfies the
\emph{$(R_{BF},r, m, \alpha)$-staging property}.
\end{inparaenum}

We start with the following lemma
that shows that according to the revealing processes $R_{BF}$ and $R_{AT}$,
the edge probabilities are in fact equal in ${\rm PA}_{p,m}(n)$.

\begin{lemma} \label{lem:two-histories}
Let $L$ be a graph generated from the ${\rm PA}_{p,m}(n)$ model but with the $j$-th edge issued by vertex $v$ missing.  Then in the PA-model, $p_{(u, v, j), L, AT} = p_{(u, v, j), L, BF}$.
\end{lemma}
\begin{proof}
First, we inspect $p_{(u, v, j), L, AT}$. At the time of processing $(u, v, j)$, $R_{AT}(\mathcal{G})$ has revealed the degree of $u$ and the sum of degree of all nodes before $u$. We also know that the $j$-th
edge did not connect to any node with indices greater than $u$. Conditioned
on these, we have
\begin{align*}
p_{(u, v, j), L, AT}=\frac{p' d(u)}{\sum^u_{i=1} d(i)}+\frac{1-p'}{u},
\quad \text{where } p' = \frac{p\frac{\sum_{u' \leq u} d(u')}{\sum_{v' \leq v} d(v')}}{p\frac{\sum_{u' \leq u} d(u')}{\sum_{v' \leq v} d(v')} + (1 - p)\frac{u}{v}}.
\end{align*}
Note that $p'$ is the probability that we choose the $j$-th edge of $v$ preferentially.  We must update this using Bayes' theorem because the probability that the $j$-th edge is chosen preferentially can change conditioned on the fact that it is not attached to later arriving nodes.

As for $p_{(u, v, j), L, BF}$, $R_{BF}(\mathcal{G})$ has revealed the degree of $u$ and the sum of degree of all
nodes before $u$. However, there is extra information in the revealed filter.
There is information about edges landing on nodes after $v$ (with bigger indices), and
there is information about the number of edges that go from nodes with bigger
indices than $v$ to nodes with smaller indices than $u$. However, since the filter
hasn't revealed the degree of nodes before $u$, the filter contains no information about the \emph{distribution} of
these ``dangling'' edges. The preferential attachment is oblivious to edges which
landed after $v$. Hence, if there is no information about the distribution of the dangling edges,
the probability of $p_{(u, v, j), L, BF}$ is independent of the extra information in the filter:
\begin{align*}p_{(u, v, j), L, BF}=\frac{p' d(u)}{\sum^u_{i=1} d(i)}+\frac{1-p'}{u}, \quad
\text{where } p' = \frac{p\frac{\sum_{u' \leq u} d(u')}{\sum_{v' \leq v} d(v')}}{p\frac{\sum_{u' \leq u} d(u')}{\sum_{v' \leq v} d(v')} + (1 - p)\frac{u}{v}}.
\end{align*}
\end{proof}

%

\begin{corollary}\label{cor:two-histories}
Let $L$ be a graph generated from ${\rm PA}_{p,m}(n)$ but with the $j$-th edge out of vertex $v$ missing.  Then $p_{(W, v, j), L, AT} = p_{(W, v, j), L, BF}$.
\end{corollary}

\begin{proof}  For each revealing process, when processing $(u, v, j)$ for $u \in W$ couple the coins.  Otherwise, choose the outcome that is consistent with $L$.
\end{proof}


\begin{lemma} \label{lem:PA-isStaged}
   The ${\rm PA}_{p,m}(n)$, $m$-generates a network
   and satisfies the $(R_{BF},\log n,m,2/3)$-staging property.
\end{lemma}

\begin{proof}
That ${\rm PA}_{p,m}(n)$ is an $m$-generated network simply comes from the definition.
We define the stages as follows. Stage $S_{0}$ contains the first $2$ nodes and for each
$i$, $S_{i}=\{v_{s}|(3/2)^{i}< s \leq (3/2)^{i+1}\}$.

Let $L$ be a graph generated from ${\rm PA}_{p,m}(n)$ but with the $j$-th edge out of vertex $v$ missing. Let $W$ be the set of nodes in stage $i$ that arrived before $v$.
By Corollary~\ref{cor:two-histories}, we have that $p_{(W, v, j), L, AT} = p_{(W, v, j), L, BF}$.

We bound $p_{(W, v, j), L, AT}$. In the case that the edge of $v$ was chosen uniformly,
the probability of choosing an edge in stage $S_{i-1}$ or smaller is greater than $2/3$.
In the case that the edge was chosen preferentially, we know that the total sum of nodes before $u$ is $2m (v-1)$, and the sum of degrees for
the nodes in stage $i-1$ or smaller is at least $2m(3/2)^{i}$. Since $v<(3/2)^{i+1}$, then the probability that the preferentially selected neighbor is among the first $i-1$ stages is bigger than $2/3$. Hence $p_{(W, v, j), L, BF}=p_{(W, v, j), L, AT}<1/3$.
\end{proof}


\subsection{Speed of Complex Contagions in Graphs With Staging Property} \label{sec:arrival-prop-speed}
In this subsection, we prove that complex contagions on graphs with staging property are fast with high probability
if the initial seeds are the oldest nodes.
We show that the speed of complex contagions on graphs with staging property is bounded by the length of the longest path to the initial seeds, which is then bounded by the depth of an appropriate branching process using a coupling argument.

Our main theorem states that starting from the oldest nodes, a $k$-complex contagion on graphs with $(R_{BF}, O(\log n), k, \alpha)$-staging property (where $\alpha > 1 - 1/k$) is fast with high probability. It is noteworthy to observe that
the same scenario also happens for $k$-complex contagions on graphs with $(R_{BF},O(\log n), m, \alpha)$-staging property
where $ m\geq k$. We assume both $k$ and $m$ to be constant parameters.

\begin{theorem} \label{thm:fast-contagion-staged}
Let $\mathcal{G}(n)$ be a network that satisfies the \emph{$(R_{BF},x, k, \alpha)$-staging property} where $\alpha > 1 - 1/k$, and $x=O(\log n)$.
Also let $\mathcal{I}$ be
the set of first $k$ arrived vertices in $\mathcal{G}(n)$. A $k$-complex contagion ${\rm CC}(\mathcal{G}(n), k, \mathcal{I})$ will
infect the entire graph with probability $1 - 1/n^{c_3}$ in time $\leq c_2 \log n$ where $c_2 =(c_3 + 1+ x/(\log n \log_{kd}(e)))/\log(1/\delta)+1$, $d = k\alpha/(1 - k(1 - \alpha))$, and $\delta =  k(1-\alpha)+1/k-(1-\alpha)$.
\end{theorem}

\begin{proof}
Consider a directed subgraph of $\mathcal{G}(n)$, in which we only keep the $k$ edges from each vertex pointing to the smaller labeled vertices. We say $u$ follows $v$ if there is a directed edge from $u$ to $v$. Node $u$ becomes infected in the next round if it follows $k$ infected neighbors. By removing extra edges and making the propagation directed we only make the contagion spread slower. Thus, we get an upper bound on the speed.

We prove by induction that the time it takes to infect a vertex $v$ is no greater than the length of the \emph{longest} path from $v$ to the vertices in $\mathcal{I}$ in this directed graph. The first $k$ vertices have longest paths of length 0 to $\mathcal{I}$ and are infected at time 0.  Assume the hypothesis for nodes with path length $\ell$.  Let $\ell+1$ be the length of the longest path from a vertex $u$ to $\mathcal{I}$.  Then the $k$ out-neighbors of $u$ have paths of length at most $\ell$ to the first $k$ vertices.  By induction, they are infected at time $\ell$, and so is $u$ at time $\ell + 1$.

Pick an arbitrary node $u$.  We will show that $u$ is infected in time $O(\log n)$ with probability $1 - 1/n^{c_3+1}$.  Then taking a union bound on all nodes, we will have our result. Note that if $u$ is in stage 0, then it will be infected in time $\log n$ because stage 0 has only $\log n$ nodes and the path back to the original $k$ vertices makes progress at each step, and thus takes time at most $\log n$.
Next, we let $u$ be in stage $i > 0$. We will bound the time $t$ it takes all paths starting at $u$ to get back to stage 0, and this will bound the time to infect $u$ by $t + \log(n)$. Next, we only need to show that $t\leq (c_2 - 1) \log n$ with probability at least $1-n^{-(c_3+1)}$.

\smallskip\noindent\textbf{Coupling the longest path with the branching process.}
We will create a coupling so that the longest path from $u$ to stage $0$ is bounded by the time it takes an appropriate labeled branching process to terminate.
Let $B(y)$ denote a $B(k, i, \alpha)$-labeled branching process rooted at node $y$ (ref. to Definition~\ref{def:branching-process}).
We consider the branching process $B(\hat{u})$ that is rooted at node $\hat{u}$ labeled $i$. Node $\hat{u}$ corresponds to the node $u$ in $G$ and because $u$ is in stage $i$, $\hat{u}$ is also labeled $i$. We use the same letter to show correspondence between the branching process and the graph nodes, while node letters in $B(u)$ will carry the $\ \hat{} \ $ hat! We reveal the nodes/edges using the $R_{BF}$ process. The BF ordering determines the random choices to be revealed next.

We will couple the $j$-th branch of $\hat{u}$ to the $j$-th neighbor of $u$ in $G$. 
If the $j$-th neighbor of $u$ is NOT in stage $i$, then we couple this to 
the $j$-th branch of $\hat{u}$ so that its label is $i-1$.
This coupling is truthful to the marginal probabilities because: 
\begin{inparaenum}[\itshape a\upshape)]
\item The probability that the $j$-th edge of $u$ is in stage $i$ (over the probability of the coin flips
$\{ \phi(z,u,j)|\forall z \in S_{i} \}$) is at most $1-\alpha$ according to the staging property;
\item and the probability that $\hat{u}$ has a branch of label $i$ is $1-\alpha$ in expectation.
\end{inparaenum}

Consider a fixed node $v$ in the graph; we explain how we find the corresponding
node $\hat{v}$ in the branching process.
We wait until all the oriented edge triples $(v, w, k)$ have been revealed.
When all these triples have been revealed, we know if $v$ has:
\begin{inparaenum}[\itshape a\upshape)]
\item No corresponding parent in the branching process tree;
\item Exactly one corresponding parent $\hat{p}$ in the branching process tree;
\item More than one parent in the branching process tree.
\end{inparaenum}

We treat these cases as follows:
\begin{inparaenum}[\itshape a\upshape)]
\item We don't couple the probabilities;
\item We correspond the child of $\hat{p}$ with $v$ and name it $\hat{v}$. We couple the events as we described above;
\item We know which parent is deeper in the branching process, we couple with this branch and ignore the rest.
\end{inparaenum}  The detailed coupling procedure maintains the invariant that the label of $\hat{v}$
is always greater than the stage of the corresponding $v$ in $G$.

Lemma~\ref{lemma:main} states that the $B(k, x, \epsilon)$-labeled branching process $B(\hat{u})$ dies out after 
$(c_2 - 1)\log n$
levels with probability at least $1-n^{(-c_3+1)}$. Hence, the length of the longest path from $u$ to initial nodes is also less than
$(c_2 - 1)\log n$  with probability at least $1-n^{-(c_3+1)}$.
\end{proof}

\begin{corollary} \label{thm:main}
Let $\mathcal{I}$ be the set of first $k$ arrived vertices in the ${\rm PA}_{p,m}(n)$ graph and let $k\leq m=O(1)$. A $k$-complex contagion ${\rm CC}(G, k, \mathcal{I})$ infects the entire ${\rm PA}_{p,m}(n)$ in $O(\log n)$ rounds with high probability.
\end{corollary}

\begin{remark}
It is noteworthy that the family of graphs ${\rm PA}_{p,m}(n)$
does not always generate a power-law graph. In fact, the ${\rm PA}_{0,m}(n)$ model
generates a heavily concentrated degree distribution with the largest degree being $O(\log n)$.
We emphasize that our results about the fast and complete spread
of complex contagions hold for all the members of this family regardless of them
having a power-law distribution or not\footnote{The $G(n,p)$ graph also
has a heavily concentrated distribution with largest degree being $O(\log n)$. However,
unlike ${\rm PA}_{0,m}(n)$,
deterministic choice of a constant number of initial seeds in the $G(n,p)$
would not cause complex contagions to spread~\cite{Janson12}.}.
\end{remark}

Using the same techniques, we prove the same result about the Copy model in Appendix~\ref{app:copy}.

\begin{theorem}\label{thm:main-average}
Let $\mathcal{I}$ be the set of first $k$ arrived vertices in the ${\rm CM}_{p,m}(n)$ graph and let $k\leq m=O(1)$. A $k$-complex contagion ${\rm CC}(G, k, \mathcal{I})$ infects the entire ${\rm CM}_{p,m}(n)$ in $O(\log n)$ rounds with high probability.
\end{theorem}


\section{Conclusions and Future Work}
We proved that complex contagions in a general family of time-evolving networks (that includes the PA model) are fast if the early arriving nodes (i.e., the oldies) are infected. Without infecting the oldies, complex contagions in the PA model starting from uniformly random initial seeds (even with a polynomial number of them) would stop prematurely. These results further emphasize the importance of crucial graph structures in enabling fast and widespread complex contagions~\cite{G08,Ghasemiesfeh:2013:CCW}. Our proof techniques could also be tailored to show fast complex contagions in the copy model (Appendix~\ref{app:copy}).

As future work, it would be interesting to explore complex contagions beyond the $k$-threshold model considered in this work. Despite our complexity result that seems to preclude an exact characterization of networks that spread complex contagions, it would be interesting to create a more unified framework characterizing graph structures crucial to analyzing the speed of complex contagions.

\bibliographystyle{abbrv}
\bibliography{complex-contagion-pref-att}

\begin{thebibliography}{10}

\bibitem{Adar05}
E.~Adar and L.~A. Adamic.
\newblock Tracking information epidemics in blogspace.
\newblock In {\em Proceedings of the 2005 IEEE/ACM International Conference on
  Web Intelligence}, pages 207--214, 2005.

\bibitem{adler91bookstrap}
J.~Adler.
\newblock {Bootstrap percolation}.
\newblock {\em Physica A: Statistical and Theoretical Physics},
  171(3):453--470, Mar. 1991.

\bibitem{reka02statistical}
R.~Albert and A.-L. Barab\'asi.
\newblock Statistical mechanics of complex networks.
\newblock {\em Rev. Mod. Phys.}, 74:47--97, 2002.

\bibitem{Amini10}
H.~Amini.
\newblock Bootstrap percolation and diffusion in random graphs with given
  vertex degrees.
\newblock {\em Electr. J. Comb.}, 17(1), 2010.

\bibitem{Amini12}
H.~Amini and N.~Fountoulakis.
\newblock What {I} tell you three times is true: bootstrap percolation in small
  worlds.
\newblock In {\em Proceedings of the 8th international conference on Internet
  and Network Economics}, pages 462--474, 2012.

\bibitem{Aral09}
S.~Aral, L.~Muchnik, and A.~Sundararajan.
\newblock {Distinguishing influence-based contagion from homophily-driven
  diffusion in dynamic networks}.
\newblock {\em Proceedings of the National Academy of Sciences},
  106:21544--21549, 2009.

\bibitem{Bakshy11}
E.~Bakshy, J.~M. Hofman, W.~A. Mason, and D.~J. Watts.
\newblock Everyone's an influencer: quantifying influence on twitter.
\newblock In {\em Proceedings of the fourth ACM international conference on Web
  Search and Data Mining}, pages 65--74, 2011.

\bibitem{Bakshy09}
E.~Bakshy, B.~Karrer, and L.~A. Adamic.
\newblock Social influence and the diffusion of user-created content.
\newblock In {\em Proceedings of the 10th ACM Conference on Electronic
  Commerce}, pages 325--334, 2009.

\bibitem{BaloghP07}
J.~Balogh and B.~Pittel.
\newblock Bootstrap percolation on the random regular graph.
\newblock {\em Random Struct. Algorithms}, 30:257--286, 2007.

\bibitem{barabasi99emergence}
A.~Barab\'asi and R.~Albert.
\newblock Emergence of scaling in random networks.
\newblock {\em Science}, 286:509--512, 1999.

\bibitem{Berger05}
N.~Berger, C.~Borgs, J.~T. Chayes, and A.~Saberi.
\newblock On the spread of viruses on the {Internet}.
\newblock In {\em Proceedings of the sixteenth annual ACM-SIAM symposium on
  Discrete algorithms}, pages 301--310, 2005.

\bibitem{BergerBoCh10}
N.~Berger, C.~Borgs, J.~T. Chayes, and A.~Saberi.
\newblock Asymptotic behavior and distributional limits of preferential
  attachment graphs.
\newblock {\em Annals of Applied Probability}, 2010.

\bibitem{fitness-model}
G.~Bianconi and A.-L. Barab\'asi.
\newblock Competition and multiscaling in evolving networks.
\newblock {\em EPL (Europhysics Letters)}, 54(4):436, 2001.

\bibitem{Bol03}
B.~Bollobas.
\newblock Mathematical results on scale-free random graphs.
\newblock In {\em Handbook of Graphs and Networks}, pages 1--37. Wiley, 2003.

\bibitem{Bol04}
B.~Bollob\'{a}s and O.~Riordan.
\newblock The diameter of a scale-free random graph.
\newblock {\em Combinatorica}, 24:5--34, 2004.

\bibitem{BollobasRiSp01}
B.~Bollob\'{a}s, O.~Riordan, J.~Spencer, and G.~E. Tusnady.
\newblock The degree sequence of a scale-free random graph process.
\newblock {\em Random Struct. Algorithms}, 18(3):279--290, 2001.

\bibitem{Broder:2000}
A.~Broder, R.~Kumar, F.~Maghoul, P.~Raghavan, S.~Rajagopalan, R.~Stata,
  A.~Tomkins, and J.~Wiener.
\newblock Graph structure in the web.
\newblock In {\em Proceedings of the 9th international World Wide Web
  conference on Computer networks}, pages 309--320, 2000.

\bibitem{centola2010spread}
D.~Centola.
\newblock The spread of behavior in an online social network experiment.
\newblock {\em {Science}}, 329(5996):1194, 2010.

\bibitem{G08}
D.~Centola and M.~Macy.
\newblock {Complex Contagions and the Weakness of Long Ties}.
\newblock {\em American Journal of Sociology}, 113(3):702--734, 2007.

\bibitem{cha2010}
M.~Cha, H.~Haddadi, F.~Benevenuto, and K.~Gummadi.
\newblock Measuring user influence in {Twitter}: The million follower fallacy.
\newblock In {\em 4th International AAAI Conference on Weblogs and Social Media
  (ICWSM)}, 2010.

\bibitem{Chalupa79bootstrap}
J.~Chalupa, P.~L. Leath, and G.~R. Reich.
\newblock Bootstrap percolation on a bethe lattice.
\newblock {\em Journal of Physics C: Solid State Physics}, 12(1):L31, 1979.

\bibitem{Chieri10}
F.~Chierichetti, S.~Lattanzi, and A.~Panconesi.
\newblock Almost tight bounds for rumour spreading with conductance.
\newblock In {\em Proceedings of the 42nd ACM symposium on Theory of
  computing}, pages 399--408, 2010.

\bibitem{Chier10}
F.~Chierichetti, S.~Lattanzi, and A.~Panconesi.
\newblock Rumour spreading and graph conductance.
\newblock In {\em Proceedings of the Twenty-First Annual ACM-SIAM Symposium on
  Discrete Algorithms}, pages 1657--1663, 2010.

\bibitem{Chier11}
F.~Chierichetti, S.~Lattanzi, and A.~Panconesi.
\newblock Rumor spreading in social networks.
\newblock {\em Theoretical Computer Science}, 412(24):2602 -- 2610, 2011.

\bibitem{Coleman:1966}
J.~S. Coleman, E.~Katz, and H.~Menzel.
\newblock {\em {Medical Innovation: A Diffusion Study}}.
\newblock Bobbs-Merrill Co, 1966.

\bibitem{DK10}
E.~David and K.~Jon.
\newblock {\em Networks, Crowds, and Markets: Reasoning About a Highly
  Connected World}.
\newblock Cambridge University Press, 2010.

\bibitem{Doerr11}
B.~Doerr, M.~Fouz, and T.~Friedrich.
\newblock Social networks spread rumors in sublogarithmic time.
\newblock In {\em Proceedings of the 43rd annual ACM symposium on Theory of
  computing}, pages 21--30, 2011.

\bibitem{Domingos:2001:MNV}
P.~Domingos and M.~Richardson.
\newblock Mining the network value of customers.
\newblock In {\em Proceedings of the Seventh ACM SIGKDD International
  Conference on Knowledge Discovery and Data Mining}, pages 57--66, 2001.

\bibitem{Dommers}
S.~Dommers, R.~van~der Hofstad, and G.~Hooghiemstra.
\newblock Diameters in preferential attachment models.
\newblock {\em Journal of Statistical Physics}, 139:72--107, 2010.

\bibitem{DorogovtsevMeSa00}
S.~N. Dorogovtsev, J.~F.~F. Mendes, and A.~N. Samukhin.
\newblock Structure of growing networks with preferential linking.
\newblock {\em Phys. Rev. Lett.}, 85:4633--4636, 2000.

\bibitem{durrett06}
R.~Durrett.
\newblock {\em Random Graph Dynamics (Cambridge Series in Statistical and
  Probabilistic Mathematics)}.
\newblock Cambridge University Press, New York, NY, USA, 2006.

\bibitem{Ghasemiesfeh:2013:CCW}
G.~Ghasemiesfeh, R.~Ebrahimi, and J.~Gao.
\newblock Complex contagion and the weakness of long ties in social networks:
  revisited.
\newblock In {\em Proceedings of the fourteenth ACM conference on Electronic
  Commerce}, pages 507--524, 2013.

\bibitem{Goldschlager77}
L.~M. Goldschlager.
\newblock The monotone and planar circuit value problems are log space complete
  for p.
\newblock {\em SIGACT News}, 9(2):25--29, July 1977.

\bibitem{Gran78}
M.~Granovetter.
\newblock Threshold models of collective behavior.
\newblock {\em The American Journal of Sociology}, 83(6):1420--1443, 1978.

\bibitem{Hagberg06convergence}
O.~Hagberg and C.~Wiuf.
\newblock Convergence properties of the degree distribution of some growing
  network models.
\newblock {\em Bulletin of Mathematical Biology}, 68(6):1275--1291, 2006.

\bibitem{Harris64}
T.~E. Harris.
\newblock {\em The theory of branching processes}.
\newblock Die Grundlehren der Mathematischen Wissenschaften, Bd. 119.
  Springer-Verlag, Berlin, 1963.

\bibitem{Hesse02}
W.~Hesse, E.~Allender, and D.~A.~M. Barrington.
\newblock Uniform constant-depth threshold circuits for division and iterated
  multiplication.
\newblock {\em Journal of Computer and System Sciences}, 65(4):695 -- 716,
  2002.

\bibitem{Jackson08}
M.~O. Jackson.
\newblock {\em Social and Economic Networks}.
\newblock Princeton University Press, Princeton, NJ, USA, 2008.

\bibitem{Janson12}
S.~Janson, T.~Luczak, T.~Turova, and T.~Vallier.
\newblock Bootstrap percolation on the random graph ${G}_{n,p}$.
\newblock {\em Annals of Applied Probability}, 22(5):1989--2047, 2012.

\bibitem{Kleinberg:1999:WGM}
J.~M. Kleinberg, R.~Kumar, P.~Raghavan, S.~Rajagopalan, and A.~S. Tomkins.
\newblock The web as a graph: measurements, models, and methods.
\newblock In {\em Proceedings of the 5th annual international conference on
  Computing and combinatorics}, pages 1--17, 1999.

\bibitem{KumarRaRa00}
R.~Kumar, P.~Raghavan, S.~Rajagopalan, D.~Sivakumar, A.~Tomkins, and E.~Upfal.
\newblock Stochastic models for the web graph.
\newblock In {\em Proceedings of the 41st Annual Symposium on Foundations of
  Computer Science}, pages 57--, 2000.

\bibitem{Kumar:1999:ELK}
R.~Kumar, P.~Raghavan, S.~Rajagopalan, and A.~Tomkins.
\newblock Extracting large-scale knowledge bases from the web.
\newblock In {\em Proceedings of the 25th International Conference on Very
  Large Data Bases}, pages 639--650, 1999.

\bibitem{Kwak10}
H.~Kwak, C.~Lee, H.~Park, and S.~Moon.
\newblock What is {Twitter}, a social network or a news media?
\newblock In {\em Proceedings of the 19th International Conference on {World
  Wide Web}}, pages 591--600, 2010.

\bibitem{Lattanzi:2009:AN}
S.~Lattanzi and D.~Sivakumar.
\newblock Affiliation networks.
\newblock In {\em Proceedings of the 41st annual ACM symposium on Theory of
  computing}, pages 427--434. ACM, 2009.

\bibitem{Leskovec07}
J.~Leskovec, L.~A. Adamic, and B.~A. Huberman.
\newblock The dynamics of viral marketing.
\newblock {\em ACM Trans. Web}, 1(1), 2007.

\bibitem{Mitzenmacher_abrief}
M.~Mitzenmacher.
\newblock A brief history of generative models for power law and lognormal
  distributions.
\newblock {\em Internet Mathematics}, 1:226--251, 2004.

\bibitem{Mon09}
A.~Montanari and A.~Saberi.
\newblock Convergence to equilibrium in local interaction games.
\newblock {\em SIGecom Exch.}, 8(1):11:1--11:4, July 2009.

\bibitem{Morris97contagion}
S.~Morris.
\newblock Contagion.
\newblock {\em Review of Economic Studies}, 67:57--78, 2000.

\bibitem{Newman10}
M.~Newman.
\newblock {\em Networks: An Introduction}.
\newblock Oxford University Press, Inc., 2010.

\bibitem{Newman03thestructure}
M.~E.~J. Newman.
\newblock The structure and function of complex networks.
\newblock {\em SIAM REVIEW}, 45:167--256, 2003.

\bibitem{papadopoulos12popularity}
F.~Papadopoulos, M.~Kitsak, M.~Serrano, M.~Boguñá, and D.~Krioukov.
\newblock {Popularity versus Similarity in Growing Networks}.
\newblock {\em Nature}, 489:537--540, 2012.

\bibitem{Pennock02winnersdont}
D.~M. Pennock, G.~W. Flake, S.~Lawrence, E.~J. Glover, and C.~L. Giles.
\newblock Winners don't take all: Characterizing the competition for links on
  the web.
\newblock In {\em Proceedings of the National Academy of Sciences}, pages
  5207--5211, 2002.

\bibitem{Romero11}
D.~M. Romero, B.~Meeder, and J.~Kleinberg.
\newblock Differences in the mechanics of information diffusion across topics:
  idioms, political hashtags, and complex contagion on {Twitter}.
\newblock In {\em Proceedings of the 20th international conference on {World
  Wide Web}}, pages 695--704, 2011.

\bibitem{SteyversTenenbaum05}
M.~Steyvers and J.~B. Tenenbaum.
\newblock The large-scale structure of semantic networks: Statistical analyses
  and a model of semantic growth.
\newblock {\em Cognitive Science}, 29:41--78, 2005.

\bibitem{ugander12}
J.~Ugander, L.~Backstrom, C.~Marlow, and J.~Kleinberg.
\newblock Structural diversity in social contagion.
\newblock {\em Proc. National Academy of Sciences}, 109(16):5962--5966, April
  2012.

\bibitem{Weng10}
J.~Weng, E.-P. Lim, J.~Jiang, and Q.~He.
\newblock Twitterrank: Finding topic-sensitive influential {Twitterers}.
\newblock In {\em Proceedings of the Third ACM International Conference on Web
  Search and Data Mining}, pages 261--270. ACM, 2010.

\end{thebibliography}

\appendix

\section{Fast Complex Contagions in the Copy Model} \label{app:copy}
Although the copy model does not satisfy the staging property (Definition~\ref{def:staged}), it barely misses it. That is why, a model specific tailored argument akin to the arguments of Lemma~\ref{lem:PA-isStaged} and Theorem~\ref{thm:fast-contagion-staged} can be used to prove the same result on fast spreading of complex contagions in the copy model.
\begin{restate}{Theorem~\ref{thm:main-average}}
Let $\mathcal{I}$ be the set of first $k$ arrived vertices in the ${\rm CM}_{p,m}(n)$ graph and let $k\leq m=O(1)$. A $k$-complex contagion ${\rm CC}(G, k, \mathcal{I})$ infects the entire ${\rm CM}_{p,m}(n)$ in $O(\log n)$ rounds with high probability.
\end{restate}
\begin{proof}
The structure of the proof is the same as Theorem~\ref{thm:fast-contagion-staged}.
We will create a coupling so that the longest path from $u$ to stage $0$ is 
bounded by the time it takes an appropriate labeled branching process to terminate.
Let $B(y)$ denote a $B(k, i, 2/3)$-labeled branching process rooted at node $y$.
We consider the branching process $B(\hat{u})$ that is rooted at node 
$\hat{u}$ labeled $i$. Node $\hat{u}$ corresponds to the node $u$ in 
$G$ and because $u$ is in stage $i$, $\hat{u}$ is also labeled $i$.
Stage $S_{0}$ contains the first two nodes and for each
$i$, $S_{i}=\{v_{s}|(3/2)^{i}< s \leq (3/2)^{i+1}\}$.

We will couple the branches of $\hat{u}$ to the neighbors of $u$ in $G$.
Let $t_{0}(u)$ be first time an edge of the node $u$ is visited in the $BF$ order.
At time $t_{0}(u)$, the filter $R_{BF}$ reveals the the prototype node $z_u$ and
how many edges of $u$ are copied from the prototype and how many are chosen
uniformly randomly. The random choices specific to the node $u$ might dictate the outcome of
all or some of its random edges.
We will couple the children of $\hat{u}$ to the neighbors of $u$ following the dictated
pattern of node $u$'s specific random choices. For example, in the copy model
 if the prototype node $z_u \notin W$, we consider all its prototype outgoing edges 
 as being outside $W$.

If a neighbor of $u$ lands in stage $i-1$, the label of the corresponding child of $\hat{u}$ 
will also be $i-1$. Otherwise, the label of the child would be $i$ as its parent. 
The correspondence between the nodes of the branching process and $G$ is made later.

We handle the appearance of multiple candidates to be coupled with a node $v$
the same way that the proof of Theorem~\ref{thm:fast-contagion-staged} handles it:
We simply couple children of $v$ with the deepest candidate $\hat{v}$ in $B(u)$.

We need to show that the coupling is truthful to the marginal probabilities.
Let $x$ be an arbitrary node in the graph with stage $j$ with a corresponding node $\hat{x}$ in the branching process. 
Firstly, we know that $\hat{x}$ creates at most $m/3$ children labeled $j$ in expectation.
Let $L$ be a graph generated from ${\rm CM}_{p,m}(n)$ that is consistent
with the revealed information in $R_{BF}(\mathcal{G})$ up until time $t_0(x)$.
 
In the revealed filter, there is information about
\begin{inparaenum}
\item the prototype nodes of all the nodes with indices bigger than $x$;
\item edges landing on nodes after $x$ (with bigger indices);
\item the number of edges that go from nodes with bigger
indices than $x$ to nodes with smaller indices that $x$.
\end{inparaenum}

Let $W$ be the set of nodes in stage $j$ before $x$ in the arrival order.
We prove that conditioned on the information in $R_{BF}(\mathcal{G})$
up until time $t_0(x)$, $x$ creates at most $m/3$ edges to $W$.

Node $x$ chooses a uniformly random prototype $z_x$. Then with probability $p$, node $x$ chooses the $\ell$-th 
outgoing edge of the prototype node $z_x$ as its edge; and with probability $1-p$ 
it chooses its $\ell$-th outgoing edge uniformly at random. 
The revealed information in the filter about the edges and prototypes of other nodes
does not affect any of the random choices of $x$.
Furthermore, when triples of the type $(u,x,j)$ are processed, the prototype of $v_x$
is already revealed and cannot be changed.
Node $z_x$ would be outside $W$ with probability at least
$1/2$. Furthermore, since \emph{outgoing} edges of $z_x$ appear before $z_x$ in the arrival order,
the copied edge from the prototype node are outside $W$ with probability $\geq 2/3$.
If an edge was chosen uniformly at random, it will be outside $W$ with probability at least $2/3$.
Hence we have conditioned on the information in $R_{BF}(\mathcal{G})$
up until time $t_0(x)$, $x$ creates at most $m/3$ edges to $W$.

The above argument shows that the coupling is truthful to the marginal probabilities.
The explained coupling procedure maintains the invariant that the label of $\hat{v}$
is always greater than the stage of the corresponding vertex $v$ in $G$.
Using Lemma~\ref{lemma:main}, we conclude that the length of the longest path from $u$
to initial nodes is also less than
$(c_2 - 1)\log n$  with probability at least $1-n^{-(c_3+1)}$ for constants $c_2,c_3$ depending
on $k$.
Hence the speed of a $k$-complex contagion
is $O(\log n)$ with high probability.
\end{proof}
\section{Bootstrap Percolation in the Preferential Attachment Model} \label{sec:bootstrap}
In this section, we focus on bootstrap percolation in the Preferential Attachment model (Definition~\ref{def:PA}). In other terms, we analyze complex contagions when the initial seeds are chosen uniformly at random.  First, we show that there exists a polynomial threshold $f(n)$ such that if $o(f(n))$ initial seeds are chosen uniformly at random, the contagion almost surely does \emph{not} spread. Second, we show that if $\Omega(f(n)\log n)$ initial seeds are infected, the whole
graph gets infected with high probability in $O(\log n)$ rounds. This shows that the first few nodes in the arriving order of the network are critical in their roles of enabling a complex contagion.

\subsection{No New Infections}
First, we show that choosing initial seeds randomly in the PA graph is a 
pretty inefficient way of initiating a complex contagion. The following theorem shows that
until the size of randomly chosen initial seeds is a polynomial in the size of the graph, the
contagion almost surely does not spread to any other node.
\begin{theorem}\label{thm:bootstrap}
Consider the ${\rm PA}_{p,m}(n)$ graph. A $k$-complex contagion ${\rm CC}({\rm PA}_{p,m}(n), k, S)$ would not spread to other nodes with probability $1-o(1)$, if we choose $S$ as follows. 
\benum
\item If $k \geq 2/p$, $S= \{o \left(n^{1-p/2} \right) \text{ random initial seeds} \}$;
\item If $k<2/p$, $S= \{o \left( n^{1-1/k} \right)\text{ random initial seeds} \}$.
\eenum
\end{theorem}
\begin{proof}
Assume that the network edges are undirected and let $s=|S|$. Denote by $X$ the number of infected nodes in the first round. $X$ is the number of nodes that have at least $k$ neighbors in $S$. We will show below that the expectation of $X$ is $o(1)$. By Markov's inequality, the number of infected nodes will be zero
with probability $1-o(1)$.

Let $d_i$ and $\nu_{i}(S)$ denote the degree of the $i$-th node, and the number of neighbors of node $i$ in set $S$ respectively. The expectation of $X$ can be written as
\begin{align*}
E[X]=\sum^{n}_{i=1} \Prob{\nu_{i}(S)\geq k} 
=\sum^{n}_{i=1} \sum^{mn}_{x=k}  \Prob{\nu_{i}(S) \geq k |  d_i=x} \Prob{ d_i=x}.
\end{align*}
In the proof of Lemma~\ref{lem:applyS}, we show that $$ \Prob{\nu_{i}(S) \geq k |  d_i=x}\leq W=\Min{ \left(\frac{xs}{n} \right)^{k}\left(\frac{1}{1-xs/n} \right) ,1}$$ Take $E[N_x]$ as the expected number of nodes of degree $x$ in the PA graph of $n$ vertices,
\begin{align*}
E[X]&\leq \sum^{n}_{i=1} \sum^{mn}_{x=k} W \cdot \Prob{ d_i=x}  \leq \sum^{mn}_{x=k} W\cdot  E[N_x].
\end{align*}
Thus, a critical step in the proof is to upper bound $E[N_x]$. We utilize the \emph{master equation} method~\cite{DorogovtsevMeSa00} to perform this computation.
However, instead of directly solving the recurrence as is done
for the case of $p=1$ for the \emph{sequential} PA model
in \cite{durrett06} and for the \emph{conditioned} PA model in \cite{Hagberg06convergence},
we upper bound it for all values of $0\leq p\leq 1$ in Lemmas~\ref{lem:tight-mean} and \ref{lem:bounding-eta}.

 Let $N_t(x)$ denotes the number of nodes with degree $x$ in the graph of $t$ vertices and denote by $n_t(x)=E[N_t(x)]$. 
The following recurrence holds for the ${\rm PA}_{p,m}$ model:
\begin{align}\label{eq:recursion-expection}
E[N_{t+1}(x)|N_t(x)]=\left(1-\frac{a_x}{t} \right)n_t(x) +\frac{a_{x-1}}{t} n_t (x-1) +c_x.
\end{align}
in which $a_x$ and $c_x$ are non-negative values that depend on the specific model and $a_{x+1}\geq a_x$.

In the ${\rm PA}_{p,m}$ model, each node issues $m$ edges to existing nodes. With probability $p$, each edge connects to a node with preferential attachment rule and with probability $1-p$, an edge connects to a uniformly random chosen node.
\begin{align*}
a_x =\frac{px}{2}+m(1-p), \qquad  c_x =\delta_{(m+1)x}= 
\begin{cases} 1 & x=m+1 \\
0 & x\neq m	\end{cases}.
\end{align*} 
We ignore the possibility of more than one edge being attached to one vertex and the self-loops.
We present the rest of the proof in the following four lemmas.

\begin{lemma}\label{lem:tight-mean}
Let $N_x$ be the number of nodes of degree $x$ in the ${\rm PA_{p,m}(n)}$ model. We have that  $E[N_n(x)]\leq m n \eta_x$, where $\eta_x=\frac{a_{x-1}}{1+a_x} \eta_{x-1} +\frac{c_x}{1+a_x}$.
\end{lemma}

\begin{proof}
We prove the claim by induction on $t$, the number of nodes in the graph.
In the base case $N_0(k)=0$ for all $x$, so the claim is trivially true. 
Suppose that the claim is true for $t$, i.e., $n_t(x)\leq m t \eta_x$. And $\eta_{x-1}=(1+a_x)\eta_x/a_{x-1} -c_x/a_{x-1}$. By the recurrence we have
\[\begin{array}{ll}
n_{t+1}(x) & \leq \left(1-\frac{a_x}{t} \right)mt\eta_x +\frac{a_{x-1}}{t} mt\eta_{x-1} +c_x\\
& \leq \left(1-\frac{a_x}{t} \right)mt\eta_x +a_{x-1}m\left( (1+a_x)\eta_x/a_{x-1} -c_x/a_{x-1} \right)+c_x\\
& = m(t+1)\eta_x -(m-1)c_x\\
& \leq m(t+1)\eta_x
\end{array},
\]
which proves the statement.
\end{proof}

\begin{lemma}\label{lem:bounding-eta}
In the ${\rm PA_{p,m}(n)}$ model, we have $\eta_x=\Theta \left (x^{-(1+2/p)} \right)$ for all $0 \leq p \leq 1$.
\end{lemma}

\begin{proof}
The statement for $p=1$ is proved in \cite{Hagberg06convergence}. We follow a similar
strategy to prove it for all the values of $0 \leq p<1$.
From the recursive definition of $\eta_x$, we can write:
\begin{align*}
\eta_x=\sum^{x}_{j=1} \frac{c_j}{1+a_j} \prod^{x}_{i=j+1} \frac{a_{i-1}}{1+a_i}
\end{align*}
However, $c_{j}=0$ for all $j>m+2$. Hence for $x \geq m+2$ we can write:
\begin{align*}
\eta_x =\eta_{m+2} \prod^{x}_{j=m+3} \frac{a_{j-1}}{1+a_j}  \nonumber 
=\eta_{m+2}\prod^{x}_{j=m+3} \frac{p(j-1)/2+m(1-p)}{1+pj/2+m(1-p)}
\end{align*}
Define $\alpha_{p}=\frac{m-mp-p/2}{p/2}$ and $\beta_p=\frac{m-mp+1}{p/2}$
and notice that for $p<1$, $-1 < \alpha_p< \beta_p$. We have:
\begin{align*}
\log(x) &= \log(\eta_{m+2})+ \sum^{x}_{j=m+3} \log \left( \frac{p(j-1)}{2}+m(1-p) \right) - \log \left( 1+\frac{pj}{2}+m(1-p) \right) \\
&=  \log(\eta_{m+2})+ \sum^{x}_{j=m+3} \log \left( 1+\frac{\alpha_p}{j} \right) - \log \left( 1+\frac{\beta_p}{j} \right)
\end{align*}
$f(x)=\log (1+x)$ is a continuous function. So by the \emph{mean value theorem} we have:
\begin{align*}
\forall j, \quad \exists \psi_j \quad \alpha_j/j<\psi_{j}<\beta_j/j, \quad f'(\psi_{j})=\frac{f(\beta_j)-f(\alpha_j)}{\beta_j -\alpha_j}
\end{align*}
Hence we get:
\begin{align*}
\log(x) = \log(\eta_{m+2})+ \sum^{x}_{j=m+3} \left ( \frac{\beta_j -\alpha_j}{j} \right) \frac{1}{1+\psi_{j}} 
=\log(\eta_{m+2})- \frac{2+p}{p} \sum^{x}_{j=m+3}  \frac{1}{j(1+\psi_{j})}
\end{align*}
Furthermore, we have that 
$$\sum^{x}_{j=m+3}  \frac{1}{j+\beta_p} \leq \sum^{x}_{j=m+3}  \frac{1}{j(1+\psi_{j})} \leq \sum^{x}_{j=m+3}  \frac{1}{j+\alpha_p};$$
which means that $\eta_x=\Theta \left(x^{-(1+2/p)} \right)$.
\end{proof}

\begin{lemma}\label{lem:bound-expected-infections}
Let $S$ be chosen as stated in Theorem~\ref{thm:bootstrap} and $X$ be the number of infected nodes in the first round 
of ${\rm CC}({\rm PA}_{p,m}(n), k, S)$. We have that, $E[X]=O \left( \frac{s^k}{n^{k-1}}\sum^{n/2s}_{x=k} x^{k-1-2/p} +n\sum^{mn}_{x=n/2s+1} 1/x^{1+2/p} \right)$.
\end{lemma}
\begin{proof}
Let $d_i$ and $\nu_{i}(S)$ denote the degree of the $i$-th node, and the number of neighbors of node $i$ in set $S$ respectively.
We have:
$$
E[X]=\sum^{n}_{i=1} \Prob{\nu_{i}(S)\geq k} 
=\sum^{n}_{i=1} \sum^{mn}_{x=k}  \Prob{\nu_{i}(S) \geq k |  d_i=x} \Prob{ d_i=x}.
$$
We can rewrite $\Prob{\nu_{i}(S)\geq k |d_i=x}$ as:
\begin{align*} 
\Prob{\nu_{i}(S)\geq k |  d_i=x} &= \Min{ \sum^{x}_{j=k} \Prob{\nu_{i}(S)=j |  d_i=x},1} \\
&\leq \Min{ \sum^{x}_{j=k} x^j\left(\frac{s}{n} \right)^{j} ,1}\\
&\leq \Min{ \left(\frac{xs}{n} \right)^{k}\left(\frac{1}{1-xs/n} \right) ,1} \qquad \text{if  $\frac{xs}{n}<1$}.
\end{align*}
We claim that if $\frac{xs}{n}<1/2$, then 
$$\left(\frac{xs}{n} \right)^{k}\left(\frac{1}{1-xs/n} \right)<1 \qquad \text{since $\left(\frac{1}{1-xs/n} \right) <2$ and $k\geq 2$}.$$
Now we proceed to compute an upper bound for $E[X]$:
\begin{align*}
E[X]&=\sum^{n}_{i=1} \sum^{mn}_{x=k} \Prob{\nu_{i}(S) \geq k |  d_i=x} \Prob{ d_i=x} \\
&\leq \sum^{n}_{i=1} \sum^{mn}_{x=k} \Min{ \left(\frac{xs}{n} \right)^{k}\left(\frac{1}{1-xs/n} \right) ,1}  \Prob{ d_i=x} \\
&\leq \sum^{mn}_{x=k} \Min{ \left(\frac{xs}{n} \right)^{k}\left(\frac{1}{1-xs/n} \right) ,1} \sum^{n}_{i=1}   \Prob{ d_i=x} \\
&\leq \sum^{mn}_{x=k} \Min{ \left(\frac{xs}{n} \right)^{k}\left(\frac{1}{1-xs/n} \right) ,1} E[N_x].
\end{align*}
Now we cut off the summation at $xs/n=1/2$. Although this cut-off is not sharp, since
we are bounding the expectation from above it is ok.
\begin{align*}
E[X]&\leq \sum^{mn}_{x=k} \Min{ \left(\frac{xs}{n} \right)^{k}\left(\frac{1}{1-xs/n} \right) ,1} E[N_x] \\
&\leq \sum^{n/2s}_{x=k} \left(\frac{xs}{n} \right)^{k}\left(\frac{1}{1-xs/n} \right) E[N_x] +\sum^{mn}_{x=n/2s+1} E[N_x] \\
&\leq \sum^{n/2s}_{x=k} 2 \left(\frac{xs}{n} \right)^{k} E[N_x] +\sum^{mn}_{x=n/2s+1} E[N_x]  \qquad \text{since $\left(\frac{1}{1-xs/n} \right) <2$ in the first sum}, \\
&\leq \sum^{n/2s}_{x=k} 2 \left(\frac{xs}{n} \right)^{k} mn\eta_x + \sum^{mn}_{x=n/2s+1} mn\eta_x
\qquad \text{Using Lemma~\ref{lem:tight-mean},} \\
&\leq \sum^{n/2s}_{x=k} 2 \left(\frac{xs}{n} \right)^{k} mn \, \Theta \left (x^{-(1+2/p)} \right) + \sum^{mn}_{x=n/2s+1} mn \, \Theta \left (x^{-(1+2/p)} \right)
\qquad \text{Using Lemma~\ref{lem:bounding-eta},} \\
&=O \left( \frac{s^k}{n^{k-1}}\sum^{n/2s}_{x=k} x^{k-1-2/p} +n\sum^{mn}_{x=n/2s+1} x^{-1-2/p} \right)   
\end{align*}
\end{proof}

\begin{lemma}\label{lem:applyS}
Let $S$ be chosen as stated in Theorem~\ref{thm:bootstrap}, and $X$ be the number of infected nodes in the first round 
of ${\rm CC}({\rm PA}_{p,m}(n), k, S)$. We have that $E[X]=o(1)$.
\end{lemma}
\begin{proof}
We just need to do case analysis on $E[X]$ based on Lemma~\ref{lem:bound-expected-infections}:
\bitem
\item If $k> 2/p$, then we have $E[X]=O \left( \frac{s^k}{n^{k-1}} \left( \frac{n}{2s}\right)^{k-2/p} + n \left(\frac{n}{2s}+1 \right)^{-2/p} \right)$. Thus $E[X]=O \left( \frac{s^{2/p}}{n^{2/p-1}} \right)$.

If $s=o \left(n^{1-p/2} \right)$, we get $E[X]=o(1)$.

\item If $k=2/p$, we have $E[X]=O \left( \frac{s^k}{n^{k-1}} \log (n/2s) + n \left(\frac{n}{2s}+1 \right)^{-2/p} \right)$,
which solves to $E[X]=O \left( \frac{s^{2/p}\log (n/2s)}{n^{2/p-1}} \right)$.

If $s=o \left(n^{1-p/2} (\log n)^{-p/2} \right)$, we get that $E[X]=o(1)$. However if $n^{1-p/2} (\log n)^{-p/2}< s =o \left(n^{1-p/2} \right)$, the $\log (n/2s)$ term in $E[X]$ would be a constant and $E[X]=o(1)$.

\item And if $k<2/p$, we have $E[X]=O \left( \frac{s^k}{n^{k-1}}  +  \frac{s^{2/p}}{n^{2/p-1}} \right) $
that solves to $E[X]=O \left( \frac{s^k}{n^{k-1}} \right).$
If $s=o \left( n^{1-1/k} \right)$, we get that $E[X]=o(1)$ again.
\eitem
\end{proof}

Applying Markov inequality on the statement of Lemma~\ref{lem:applyS} proves the statement of Theorem~\ref{thm:bootstrap}.
\end{proof}

\subsection{Oldies But Goodies}
We utilize the expected degree of early nodes in the PA model to show that
they become infected with high probability once enough random seeds are infected at round $0$.
Once all the first $k$ nodes in the graph are infected, the $k$-complex contagion will spread
to the rest of the graph and it spreads quickly. Once again, this emphasizes
the role of early nodes in the PA model.

A computation of the expected degree of nodes for $p=1$ and $m=1$ is presented in
\cite{BollobasRiSp01}. We follow their approach and prove
the expected degree for all values of $0 \leq p \leq 1, m \geq 2$ in the following Lemma.
We will work with the independent model here, but the other two variations
are similar.
\begin{lemma} \label{lem:degree-early-nodes}
Let $d_{t}(s)$ denote the degree of node $s$ in a ${\rm PA}_{p,m}$ at time $t$.
We have $E[d_{n}(s)] =\Theta \left ( \left(n/s \right)^{p/2} \right)
$.
\end{lemma}
\begin{proof}
We start by writing a recursive relation based on the edge probabilities.
\begin{align*}
E[d_{t}(s)|d_{t-1}(s)]&=d_{t-1}(s)+ pm\frac{d_{t-1}(s)}{2m(t-1)} +(1-p)m\frac{1}{t-1} \\
E[d_{t}(s)]&=\frac{2t-2+p}{2t-2}E[d_{t-1}(s)] +\frac{m(1-p)}{t-1}
\end{align*}
Starting with $d_s(s)=m$, we get:
\begin{align*}
E[d_{n}(s)]&= \sum^{n}_{j=s} \frac{m(1-p)}{j-1} \prod^{n}_{i=j+1} \frac{t-1+p/2}{t-1} \\
&=\sum^{n}_{j=s} \frac{m(1-p)}{j-1} d_s(s)\prod^{n}_{i=j+1} \frac{t-1+p/2}{t-1} \\
&= \sum^{n}_{j=s} \frac{m^2(1-p)}{j-1} \frac{\Gamma(n-1+p/2)}{\Gamma(n-1)} \frac{\Gamma(j)}{\Gamma(j+p/2)}    \\
&=  m^2(1-p)\sum^{n}_{j=s} \frac{1}{j-1} \left(\frac{n}{j} \right)^{p/2} \left(1+O \left (\frac{1}{j} \right) \right)  
\qquad \text{using Stirling's formula for $\Gamma(.)$}; \\
&= \Theta \left(  m^{2}(1-p)n^{p/2}\sum^{n}_{j=s} \frac{1}{j^{p/2}(j-1)} \right) \\
&= \Theta \left(  m^{2}(1-p)\left(\frac{n}{s} \right)^{p/2}  \right)
\end{align*}
\end{proof}

\begin{theorem}
If we choose $\mathcal{I}= \{ \Omega \left(n^{1-p/2} \log n \right) \text{ random initial seeds}\}$, then a $k$-complex contagion ${\rm CC}({\rm PA}_{p,m}(n), k, \mathcal{I})$ spreads to all the nodes with high probability in $O(\log n)$ rounds.
\end{theorem}
\begin{proof}
We focus on the first  $k$ arrived nodes in the PA graph. By Lemma~\ref{lem:degree-early-nodes}, each of the first $k$ nodes have
expected degree of at least $m^{2}(1-p)\left(\frac{n}{k} \right)^{p/2}$. We focus on node $v_k$, the node that arrived at time $k$, from now on. If we infect $\Omega \left(n^{1-p/2} \log n \right)$ nodes, $v_k$ would have $\Omega(\log n)$ infected neighbors in expectation in round $0$.
This would mean that with high probability $v_k$ would have $\geq k$ infected neighbors in round $0$.
This means that all the first $k$ nodes will be infected with high probability in round $1$. Once the first $k$ nodes are infected, Corollary~\ref{thm:main} can be applied to show that the speed of contagion is $O(\log n)$ whp.
\end{proof}

%

%





\section{Complexity of Computing The Extent of Complex Contagions}\label{sec:complexity}
In this section, we present a computational complexity result regarding the computation
of complex contagions.
We show that it is $\Po$-complete to decide if a $k$-complex contagion completely infects
a graph or stops at a small fraction of its nodes. The reduction comes from the \defn{MonotoneCircuitValue}
problem in circuit complexity.

\begin{definition}  In the \defn{MonotoneCircuitValue} \textsf{(MCV)} problem we are given a circuit $C$ with \texttt{0}, \texttt{1}, \texttt{AND}, and \texttt{OR} gates and one gate $g_*$ designated as $output$.  We insist that $C$ is layered, that is we can partition the gates into levels $\{\texttt{0}, \texttt{1}\} = L_0, L_1, \ldots, L_{\ell - 1},  L_{\ell} = \{g_*\}$ such that wires always connect gates at levels $i$ and $i+1$ for some $0 \leq i \leq \ell -1$.
$C \in \textsf{MCV}$ if the circuit is a properly encoded, layered, monotone circuit and evaluates to \texttt{1}.  Otherwise $C \not\in \textsf{MCV} $.
\end{definition}

\begin{theorem} [From \cite{Goldschlager77}] The \textsf{MCV} problem is $\Po$-complete.
\end{theorem}

\begin{theorem}  \label{thm:complexity}
For any integer $k \geq 2$, given a triple $(G, S, M)$ where $G$ is an undirected graph, $S$ is a subset of vertices, and $M$ is an integer,  it is $\Po$-complete to determine if the size of the resulting $k$-complex contagion on $G$ when the vertices of $S$ are initially infected is at least $M$. Let $n$ be the number of vertices in $G$. In fact for any $0 < \epsilon < 1$, the promise problem of deciding if the size of the resulting $k$-complex contagion is $n$ or at most  $n^{\epsilon}$, is promise $\Po$-complete.
\end{theorem}

\begin{proof}
These problems are in $\Po$ or promise-$\Po$ because an algorithm can simply simulate the contagion and count the number of infected nodes.  To show the hardness result the idea is to reduce from {\textsc MCV}.  Given such a circuit $C$ we create a graph as follows:

Fix $\epsilon, k$.   Given a circuit $C$ with $m$ gates we create the triple $(G, S, M)$ as follows:

Let $M = (3 k^3 m)^{1/\epsilon}$.

We next create the vertices of $G$:
\begin{itemize}
\item  For each gate $g_a$ of $C$ , we create $k$ vertices $G_a = \{ g_a^i\}_{0 \leq i < k}$.
\item  For each wire $w_{ab}$ of $C$ connecting gate $g_a$ to gate $g_b$, create $k^2$ vertices $W_{ab} = \{w_{ab}^{i, j}\}_{0 \leq i, j < k}$.
\item  Create $M$ additional vertices $T = \{t_i\}_{0 \leq i < M}$.
\end{itemize}


Next, we create the edges:
\begin{itemize}
\item  Consider a non constant gate $g_c$ of $C$ with input gates $g_a$ and $g_b$ (assume an arbitrary ordering over the input gates).
\begin{itemize}
   \item Add the $k^3$ edges to connect all vertices in $G_x$ to all vertices in $W_{xc}$ for $x \in \{a, b\}$.
   \item If $g_c$ in an \texttt{OR} gate, connect  $w_{x,c}^{i, j}$ to $g_c^i$ for $0 \leq i, j < k$ for $x \in \{a, b\}$.
   \item If $g_c$ in an \texttt{AND} gate, connect  $w_{a,c}^{i, j}$ to $g_c^i$ for $0 \leq i < k$ and $0 \leq  j < \lceil k/2 \rceil$.
   \item If $g_c$ in an \texttt{AND} gate, connect   $w_{b,c}^{i, j}$ to  $g_c^i$ for $0 \leq i < k$ and $0 \leq  j < \lfloor k/2 \rfloor$.
\end{itemize}
\item Add the $k^2$ edges between $G_*$ and $t_i$ for $0 \leq i < M$.
\item For all the vertices $v \in G \setminus T$, add $k$ edges between $v$ and $k$ vertices of $T$. 
But, each vertex of $M$ can only be used once.  Let $R = 3 k^2 m$.  Because every gate has at most 2 
in-wires, and each gate/wire has at most $k^2$ corresponding nodes, $R$ is an upper 
bound on the number of vertices not in $T$. Therefore, $M = (3 k^3 m)^{1/\epsilon}>3 k^3 m=kR$ is
big enough to satisfy the use-once constraint on the vertices of $T$.
\end{itemize}

Let $S = G_{\texttt{1}}$, the vertices corresponding to the constant \texttt{1} gate.

It is easy to verify that $(G, S, M)$ can be constructed in logspace\footnote{Note that multiplication, powering, and division are known to be in logspace~\cite{Hesse02}.  However, these results are not needed if we simply compute a number $M >(3 k^3 m)^{1/\epsilon}$}.

Now, we will show that $T$ is infected if and only if $C$ evaluates to 1.
The proof will follow from the following lemma:

\begin{lemma}
Consider a $t$, where $0 \leq t \leq 2 \ell$.  If $t$ is even, we claim that the only newly infected nodes at time $t$ correspond to gates at level $t/2$ in $C$ which evaluate to \texttt{1}.  If $t$ is odd, we claim that the only newly infected nodes at time $t$ correspond to the wires $w_{ab}$ connecting gates at level  $(t-1)/2$ and  $(t+1)/2$ where the gate at level $(t-1)/2$ evaluates to \texttt{1}.
\end{lemma}

Using the lemma, at time $2 \ell$ the only nodes that can possibly become infected are those corresponding to the output gate.  If they do become infected, then at time $2t+1$ all the nodes of $T$ will become infected. Ultimately, at time $2t+2$ all the graph will become infected.

Notice that each node in $T$ only has one edge outside the nodes of output gate $G_{*}$.
Therefore, if at time $2 \ell$ the output gate does not become infected, then at that step no additional nodes become infected and the contagion is over.

We prove the lemma first:

\begin{proof}
The proof proceeds by induction. At time $t = 0$ this is true, because the only nodes at level 0 are constant gates, and the only constant gates that evaluate to \texttt{1} is the \texttt{1} gate. By construction $G_{\texttt{1}} = S$ and so these vertices are initially infected at time $t = 0$.

Assume that the statement is true up to time $t < 2 \ell$.  We will show that the statement is true at time $t+1$.

\paragraph*{The case where $t$ is even}
At the next step, time $t+1$, any node that becomes infected must be connected to a node that was infected at time $t$.   By the inductive hypothesis, the only nodes that become infected at time $t$ are those that correspond to gates at level $t/2$.  By construction,  these nodes are connected to nodes corresponding to wires connecting gates  at level $t/2-1$ and level $t/2$ as well as nodes corresponding to wires connecting gates at level $t/2$ and level $t/2+1$.

The nodes $W_{ab}$ that correspond to wires $w_{ab}$ connecting a gate $g_a$ at level $t/2-1$ and a gate $g_b$ level $t/2$ are, by construction, attached to the nodes $G_a$ and the nodes $G_b$.  The nodes of $W_{ab}$ can only be infected at time $t+1$ if they were not already infected at time $t$.  By the inductive hypothesis, the nodes of $W_{ab}$ are not infected at time $t$ if and only if $g_a$ evaluates to \texttt{0} in which case, again by the inductive hypothesis, the nodes of $G_a$ are also not infected at time $t$.  However, if the nodes corresponding to $G_a$ are not infected at time $t$, then the nodes in $W_{ab}$ will not be infected at time $t+1$ as, by construction, each node in $W_{ab}$ has only one neighbor outside of $G_a$ and $k \geq 2$.

\paragraph*{The case where $t$ is odd}
At time $t+1$, any node that becomes infected must be connected to a node that was infected at time $t$.   By the inductive hypothesis, the only nodes that become infected at time $t$ are those that correspond to wires that connect nodes in level $(t-1)/2$ and level $(t+1)/2$.  By construction,  these nodes are connected to nodes corresponding to gates at level $(t-1)/2$ and level $(t+1)/2$.  By the inductive hypothesis, all the neighbors that these newly infected nodes' wires connect to at level $(t-1)/2$ are already infected.  Let's consider then the nodes corresponding to gates at level $(t+1)/2$.

If $g_c$ is an \texttt{OR} gate with inputs $g_a$ and $g_b$, then, by construction, each node in $G_c$ is attached to $k$ nodes in $W_{ac}$ and $k$ nodes in $W_{bc}$.  Thus, if either $g_a$ or $g_b$ evaluate to 1, then, by the inductive hypothesis, either the nodes in  $W_{ac}$ or the nodes in $W_{bc}$ will be infected at time $t$ and thus at time $t+1$ the nodes in $G_c$ will become infected.   On the other hand, if neither $g_a$ or $g_b$ evaluate to \texttt{1}, then, by the inductive hypothesis, neither the nodes in  $W_{ac}$ or the nodes in $W_{bc}$ will be infected at time $t$.  By construction, any other neighbors of nodes in $G_c$ correspond to wires connecting gates at level $(t-1)/2$ and $(t+1)/2$.  By the inductive hypothesis, these gates are not infected at time $t$.  Thus, the nodes of $G_c$ will not be infected at time $t+1$.

If $g_c$ is an \texttt{AND} gate with inputs $g_a$ and $g_b$, then, by construction, each node in $G_c$ is attached to $\lceil k/2 \rceil$ nodes in $W_{ac}$ and $\lfloor k/2 \rfloor$ nodes in $W_{bc}$.  Thus, if both $g_a$ or $g_b$ evaluate to 1, then, by the inductive hypothesis, the nodes in  $W_{ac}$ and the nodes in $W_{bc}$ will be infected at time $t$ and thus at time $t+1$ the nodes in $G_c$ will become infected.   On the other hand, if either $g_a$ or $g_b$ evaluate to \texttt{0}, then, by the inductive hypothesis, either the nodes in  $W_{ac}$ or the nodes in $W_{bc}$ will be not infected at time $t$.  By construction, any other neighbors of nodes in $G_c$ correspond to wires connecting gates at level $(t-1)/2$ and $(t+1)/2$.  By the inductive hypothesis, these gates are not infected at time $t$.  Thus, the nodes of $G_c$ will not be infected at time $t+1$.

\end{proof}
The reduction is complete because:
\begin{itemize}
\item Thus, if $C \in \textsf{MCV}$, then $T$ becomes infected and at least $M$ nodes (and all the nodes in the graph) are infected.
\item If $C \not\in \textsf{MCV}$, then $T$ does not become infected. Remember that $R=3k^2 m$ is an upper bound on the number of vertices not in $T$ and thus an upper bound on the number of nodes that become infected.  But $R=3k^2 m<M^{\epsilon}=3k^3m <n^{\epsilon}$.  Thus, fewer than  $n^{\epsilon}$ nodes are infected.
\end{itemize}
\end{proof}

\end{document}